\documentclass[a4paper]{article}
\usepackage{amsmath,amsfonts,amsthm}
\usepackage[all,cmtip]{xy}
\usepackage{hyperref}

\newtheorem{thm}{Theorem}
\newtheorem{prop}[thm]{Proposition}
\newtheorem{lem}[thm]{Lemma}

\theoremstyle{remark}
\newtheorem{rem}{Remark}

  \newcommand{\del}{\partial}
  \newcommand{\oo}{\infty}
\renewcommand{\d}{\mathrm{d}}
\renewcommand{\dh}{\mathrm{d}_{\mathsf{h}}}
  \newcommand{\dv}{\mathrm{d}_{\mathsf{v}}}
  \newcommand{\Dv}{\delta_{\mathsf{v}}}
  \newcommand{\Dh}{\delta_{\mathsf{h}}}
  \newcommand{\sso}{\subset}
  \newcommand{\sse}{\subseteq}
  \newcommand{\dR}{\mathrm{dR}}
  \newcommand{\id}{\mathrm{id}}
  \newcommand{\Secs}{\Gamma}
  \newcommand{\hSecs}{\hat{\Gamma}}
  \newcommand{\tSecs}{\tilde{\Gamma}}
  
  \newcommand{\tsxi}[1]{\tilde{\xi}^{{#1}*}}
  \newcommand{\E}{\mathcal{E}}
  \newcommand{\F}{\mathcal{F}}
  \newcommand{\G}{\mathcal{G}}
  \newcommand{\g}{\mathfrak{g}}
\renewcommand{\H}{\mathcal{H}}
\renewcommand{\S}{\mathcal{S}}
  \newcommand{\Z}{\mathcal{Z}}
  \newcommand{\eps}{\varepsilon}
  \newcommand{\hOmega}{\hat{\Omega}}
  \newcommand{\tOmega}{\tilde{\Omega}}
  \newcommand{\Char}{\mathrm{char}}
  \newcommand{\cosym}{\mathrm{cosym}}
  \newcommand{\Def}{\mathrm{def}}
  \newcommand{\cur}{\mathrm{cur}}
  \newcommand{\lin}{\mathrm{lin}}
  \newcommand{\Null}{\mathrm{null}}
  \newcommand{\src}{\mathrm{src}}

\title{Topology, rigid cosymmetries and linearization instability in
	higher gauge theories}
\author{Igor Khavkine\\
Institute for Theoretical Physics, Utrecht, Leuvenlaan 4,\\
NL-3584 CE Utrecht, The Netherlands\\ \texttt{i.khavkine@uu.nl}}

\begin{document}
\maketitle

\begin{abstract}
We consider a class of non-linear PDE systems, whose equations possess
Noether identities (the equations are redundant), including
non-va\-ri\-a\-tion\-al systems (not coming from Lagrangian field
theories), where Noether identities and infinitesimal gauge
transformations need not be in bijection. We also include theories with
higher stage Noether identities, known as higher gauge theories (if they
are variational). Some of these systems are known to exhibit
linearization instabilities: there exist exact background solutions
about which a linearized solution is extendable to a family of exact
solutions only if some non-linear obstruction functionals vanish.  We
give a general, geometric classification of a class of these
linearization obstructions, which includes as special cases all known
ones for relativistic field theories (vacuum Einstein, Yang-Mills,
classical $N=1$ supergravity, etc.). Our classification shows that
obstructions arise due to the simultaneous presence of rigid
cosymmetries (generalized Killing condition) and non-trivial de~Rham
cohomology classes (spacetime topology). The classification relies on a
careful analysis of the cohomologies of the on-shell Noether complex
(consistent deformations), adjoint Noether complex (rigid cosymmetries)
and variational bicomplex (conserved currents). An intermediate result
also gives a criterion for identifying non-linearities that do not lead
to linearization instabilities.
\end{abstract}

\section{Introduction}\label{sec:intro}
It is well known that, on spatially compact manifolds, the solution
space of Lagrangian gauge theories like General Relativity (GR) and
Yang-Mills (YM) theory is an infinite dimensional manifold with
singularities~\cite{struc1,struc2}. The same phenomenon can occur in the
solution space of other non-linear systems of partial differential
equations (PDE systems), be they Lagrangian field theories or not. A
background solution is said to be \emph{linearization stable} if a
solution space neighborhood of it can be modeled on the vector space of
linearized solutions about that background. That is, for every
linearized solution, there exists a 1-parameter family of exact
solutions tangent to it. Otherwise, the background is said to be
\emph{linearization unstable}, in which case there exist linearized
solutions not tangent to any 1-parameter family of exact solutions.

In the case of GR, the solutions space singularities at linearization
unstable backgrounds are of conical type. That is, the corresponding
solution space neighborhood can be modeled on the zero set of a
quadratic constraint on the space of linearized solutions. The presence
of such constraints is intimately linked with the spatial compactness of
the underlying manifold. For spatially non-compact manifolds with
asymptotically flat boundary conditions, conical singularities are
absent~\cite{cb-deser,cbfm}.

Based on a remarkable observation of
Moncrief~\cite{moncrief-killing1,moncrief-killing2}, the necessary and
sufficient conditions for linearization stability (resp.\
instability) have been identified with the absence (resp.\ presence) of
Killing vectors on the background~\cite{struc1,struc2}. Similar results
have been obtained for other gauge theories, including
Yang-Mills~\cite{moncrief-ym}, Einstein-Yang-Mills~\cite{arms-eym} and
supergravity~\cite{bao-sugra}. In each case, the notion of a Killing
vector needed to be generalized to express the sufficient conditions for
linearization instability.

In this note, we generalize the Killing condition to a general class of
Lagrangian gauge theories, and even to non-Lagrangian theories, that
satisfy non-trivial \emph{Noether identities},%
	\footnote{\label{noether-term}%
	This terminology deserves a comment.  Without question, the
	terminology \emph{Noether identity}~\cite{sarda,kls} is correct and
	standard for Lagrangian systems, where they have also been called
	\emph{gauge identities}~\cite{barnich-brandt}. On the other hand, for
	non-Lagrangian systems, differential identities that annihilate a
	given linear differential operator have been traditionally called
	\emph{compatibility operators}~\cite{goldschmidt-lin,verbov1,verbov2}
	or less commonly \emph{Janet operators}~\cite{pommaret}. So, what we
	will later call the \emph{Noether complex} can also be referred to as
	the \emph{compatibility complex} or \emph{Janet sequence}. We choose to
	use the name Noether identity even for non-Lagrangian systems simply
	because of the strong historical link between linearization stability
	analysis and Lagrangian field theories.} %
which mean that the field equations are redundant). In Lagrangian gauge
theories, the Noether identities are precisely dual to generators of
gauge transformations, by Noether's theorem. In non-Lagrangian theories,
the two need not be connected~\cite{kls}. We then show how these
theories acquire potential linearization instabilities, as in the above
examples. In fact, this shows that it is the Noether identities that are
responsible, rather than the gauge symmetries. Thus, since the classical
Killing condition is associated with symmetries, we call the generalized
condition responsible for linearization instabilities \emph{co-Killing}.
Since, as already mentioned, Noether's second theorem links gauge
symmetries and Noether identities, this distinction may have been
ambiguous in the past. Briefly, fields satisfying the co-Killing
condition are those in the kernel of a adjoint linearized Noether
operator and are called \emph{cosymmetries}~\cite{serg}. In specific
cases, they have also been called \emph{reducibility
parameters}~\cite{barnich-brandt}.

Starting with the original work of Fischer~\&\ Marsden on vacuum
GR~\cite{fima-linstab}, the conditions for linearization instability
have often been detected through an analysis of the constraints imposed
by the field equations on the initial data. In 4-dimensional GR, this
requires the well-known 3+1 ADM decomposition. However, given that the
resulting conditions (presence or absence of Killing vectors) are
spacetime covariant it is desirable to have a covariant derivation of
the standard and generalized co-Killing conditions as well. Such
derivations are indeed available for specific examples as for instance
for vacuum GR in~\cite{struc1} (see also~\cite{unruh-losic}, where
similar ideas appear outside the specialized literature). In this paper,
we give a covariant derivation of the generalized co-Killing condition,
applicable to the same large class of theories mentioned above.

After having completed this work we learned of a little known paper of
Arms~\& Anderson~\cite{aa-taub}, which also gives a fully covariant
derivation of generalized Killing conditions, in fact using ideas
analogous to those presented below. However, their results are rather
less general than ours and are presented in terms of explicit
calculations on the example of Einstein-Yang-Mills theory (though with
indications of how they are to be generalized). We give a brief
comparison of their results with ours at the end of this section.

Since we analyze a large class of field theories instead of working on
specific examples, we require an adequate abstract framework to state
the necessary hypotheses and carry out the analysis. The abstract
framework is provided by the jet-bundle formalism and various associated
differential complexes and their cohomologies. One benefit of the
abstraction is to attract attention to the fact that the necessary
hypotheses are surprisingly liberal. In particular, the class of
admissible PDE systems is larger than the determined elliptic and
hyperbolic systems or gauge theories closely resembling those. Also, it
becomes clear how the topology of the spacetime manifold gives rise to
linearization instabilities and that any non-trivial de~Rham cohomology
class may be responsible (not only one dual to a compact Cauchy surface).

The jet-bundle formalism is introduced in Section~\ref{sec:loc-geom};
its contents are standard and serve mostly to fix notation. The main
hypotheses imposed on the PDE systems under consideration are described
in Section~\ref{sec:pde}.  Section~\ref{sec:def-cosym-src} introduces
three important concepts: consistent deformations
(Section~\ref{sec:def}), cosymmetries and the co-Killing condition
(Section~\ref{sec:cosym}), and null sources (Section~\ref{sec:src}).
Section~\ref{sec:lin-stab} contains the main results of the paper.
Linearization instabilities and obstructions are defined
in~\ref{sec:obstr}. Special conserved currents, deformation currents,
are defined by pairing consistent deformations and null sources in
Section~\ref{sec:def-cur}. Our main result is obtained in
Section~\ref{sec:def-cur}: the deformation current defined by the
non-linearity gives rise to a linearization obstruction valued in the
de~Rham cohomology of the spacetime manifold. The relation of this
result to previous work is briefly discussed in
Section~\ref{sec:charge}, while Section~\ref{sec:asymp} discusses
obstructions generated by non-trivial asymptotic boundary conditions
rather than non-trivial topology.  Finally, the general result is
specialized to several examples, some of which are novel, in
Section~\ref{sec:ex} and Section~\ref{sec:discuss} concludes with a
discussion.

Since analogous ideas had previously appeared in~\cite{aa-taub}, let us
make a quick comparison. Arms~\& Anderson used the standard first and
second Noether theorems for Einstein-Yang-Mills equations to link
`symmetries' and `gauge symmetries' (rigid stage-$0$ and stage-$1$
symmetries, respectively, in our terminology) with conservation laws.
They perturbed these conservation laws to derive conserved currents for
the linearized equations. They called these currents `Taub forms' and
their charges (Cauchy surface integrals) `Taub numbers'. The
non-vanishing of some of the Taub numbers was then shown to be an
obstruction to linearization stability.  Thus, Taub forms are analogous
to our deformation currents paired with the equation non-linearity and
Taub numbers to our de~Rham cohomology valued obstructions. In contrast,
their work considered neither non-Lagrangian equations nor lower degree
conservation laws. Moreover, the generation of conservation laws from
the pairing of deformation currents and non-linear deformations was only
implicit in their calculations. Finally, they showed that Taub numbers
are gauge invariant. On the other hand, we do not consider this question
below, since we make no hypothesis about the presence or absence of
gauge symmetries in the class of equations that we consider.

\section{Local geometry of PDEs}\label{sec:loc-geom}
\subsection{Jets, local forms and the variational bicomplex}\label{sec:jets}
The natural setting for the local analysis of differential equations is
that of \emph{jet bundles}. The needed concepts and notation are
briefly introduced below. More details can be found in the standard
literature; see for
instance~\cite{olver,anderson-small,anderson-big,vinogradov,goldschmidt,seiler,bbh-rep}.

Given a vector bundle $F\to M$ over a connected $n$-dimensional smooth
manifold $M$, the \emph{$k$-jet bundle} $J^kF\to M$ is a vector bundle
whose defining characteristic is that for any (possibly non-linear)
differential operator $f\colon \Secs(F)\to \Secs(F')$ of order $k$,
there exists a canonical factorization $f[u] = f\circ j^k u$ for any
section $u\colon M\to F$, where the \emph{$k$-jet prolongation}
$j^k\colon \Secs(F) \to \Secs(J^kF)$ is composed with a smooth bundle
map $f\colon J^kF\to F'$, which by a slight abuse of notation we denote
using the same symbol as the original differential operator. Composing
the differential operator $f$ with an $l$-jet prolongation canonically
defines a new differential operator $p_l f\colon J^{l+k}F \to J^lF'$
called its \emph{$l$-prolongation}, $j^l f[u] = p_l f\circ j^k u$. Given
a trivializable restriction $F_U\to U$ of $F$ to a chart $U\sso M$ with
local coordinates $(x^i)$ and fiber-adapted local coordinates
$(x^i,u^a)$, there is a corresponding adapted chart $J^kF_U\sso J^kF$
with adapted local coordinates $(x^i,u^a_I)$, where $I=i_1\cdots i_l$
runs through multi-indices of orders $|I|=l=0,\ldots,k$. In these
coordinates, the $k$-jet prolongation is given by $j^k u(x) =
(x^i,\partial_I u^a(x))$, while the $l$-prolongation is given by $p_l
f[u](x) = (x^i,\del_I f^b[u](x))$, where $f[u](x) = (x^i,f^b[u](x))$ in
fiber-adapted local coordinates $(x^i,v^b)$ on $F'$. For any $l>k$,
discarding the information about all derivatives of order $>k$ defines a
natural projection $J^lF\to J^kF$.  The projective limit $J^\oo F :=
\varprojlim_{k\to\oo} J^k F$ defines the \emph{$\oo$-jet bundle}. The
$\oo$-jet prolongation $j^\oo$ and $\oo$-prolongation $p_\oo$ are
defined in the obvious way. By composing with the natural projection
$J^\oo F\to J^kF$, the differential operator $f$ also canonically defines
the smooth bundle map $f\colon J^\oo F\to J^kF \stackrel{f}{\to} F'$,
which is again denoted by the same symbol $f$. Conversely, due to the
projective limit construction, any smooth bundle map $f\colon J^\oo F
\to F'$ can only depend on finitely many coordinates of its domain,
which means that there exists a $k\ge 0$ such that this bundle map
canonically factors as $f\colon J^\oo F\to J^k F \stackrel{f}{\to} F'$,
with the smallest such $k$ being the \emph{order} of $f$.

Denote by $TM$ and $T^*M$ the tangent and cotangent bundles of $M$.
Also, let $\Lambda^k M = \bigwedge^k T^*M$ be the bundles of $k$-forms.
Denote by $\Omega^k = \Omega^k(M) = \Secs(\Lambda^k M)$ the spaces of
differential forms, with $\Omega^* = \bigoplus_k \Omega^k$. We call
$\Omega^*(J^\oo F)$ the space of \emph{local variational forms} on $F$.
The de~Rham differential on $\Omega^*(J^\oo F)$ canonically splits into
the sum $\d = \dh + \dv$, where the respective \emph{horizontal} and
\emph{vertical} differentials are individually nilpotent and
anti-commutative, $\dh^2 = \dv^2 = 0$ and $\dh\dv + \dv\dh = 0$. The
defining action of the horizontal differential on adapted local
coordinates $(x^i,u^a_I)$ is $\dh x^i = \d{x^i}$ and $\dh u^a_I =
u^a_{Ii} \d{x^i}$; then simply $\dv = \d - \dh$. Denote by
$\Omega^{h,0}(F), \Omega^{0,v}(F) \sso \Omega^*(J^\oo F)$ the subspaces
of \emph{local horizontal} and \emph{local vertical forms}, which
generate the entire space of local variational forms by wedge products.
Hence, we have a natural bigrading $\Omega^*(J^\oo F) = \bigoplus_{h,v}
\Omega^{h,v}(F)$, where $0\le h \le n$ and $0\le v < \oo$. There is a
natural inclusion $\Omega^k(M)\sso \Omega^{k,0}(F)$, via the pullback
along the natural projection $\pi_\oo \colon J^\oo F\to M$, where the
image of the inclusion is said to consist of \emph{field-independent}
forms.

The operators $\dh$ and $\dv$ together with the horizontal-vertical
bigrading turns the space of local variational forms, augmented as shown
below, into the \emph{variational bicomplex of
$F$}~\cite{anderson-small,anderson-big}:
\begin{equation}\vcenter{\xymatrix{
	& 0 \ar[d] & 0 \ar[d] & 0 \ar[d] & 0 \ar[d] \\
	0 \ar[r] &
		\underline{\Omega}^{0,0}   \ar[d]^-\dv \ar[r]^-\dh &
		\underline{\Omega}^{1,0}   \ar[d]^-\dv \ar[r]^-\dh & \cdots
		\underline{\Omega}^{n-1,0} \ar[d]^-\dv \ar[r]^-\dh &
		\underline{\Omega}^{n,0}   \ar[d]^-\dv \ar@(r,u)[rd]^-{\delta_{EL}} \\
	0 \ar[r] &
		\Omega^{0,1}   \ar[d]^-\dv \ar[r]^-\dh &
		\Omega^{1,1}   \ar[d]^-\dv \ar[r]^-\dh & \cdots
		\Omega^{n-1,1} \ar[d]^-\dv \ar[r]^-\dh &
		\Omega^{n,1}   \ar[d]^-\dv \ar[r]^-I & 
		\F^1           \ar[d]^-\Dv \ar[r] & 0 \\
	0 \ar[r] &
		\Omega^{0,2}   \ar[d]^-\dv \ar[r]^-\dh &
		\Omega^{1,2}   \ar[d]^-\dv \ar[r]^-\dh & \cdots
		\Omega^{n-1,2} \ar[d]^-\dv \ar[r]^-\dh &
		\Omega^{n,2}   \ar[d]^-\dv \ar[r]^-I & 
		\F^2           \ar[d]^-\Dv \ar[r] & 0 \\
	& \vdots & \vdots & \vdots & \vdots & \vdots
}}\end{equation}
We have abbreviated $\Omega^k = \Omega^k(M)$ and $\Omega^{h,v} =
\Omega^{h,v}(F)$. All arrows commute, which defines the
\emph{Euler-Lagrange derivative} $\delta_{EL} = I\circ \dv$, where $I$
is the so-called \emph{interior Euler operator}. The space $\F^k
\cong \bigoplus_a \dv u^a \wedge \Omega^{n,k-1}$ is called the space of
\emph{local functional $k$-forms} is defined by the image of $I$ and the
\emph{variational differential} $\delta_v$ is defined as the unique map
commuting with the rest of the diagram.

We can define the cohomologies, $H^{h,v}(\dh)$ and $H^{h,v}(\dv)$, of
the horizontal and vertical differentials in the obvious way; they
correspond to the cohomologies of the corresponding parts of the above
rows and columns. All the columns and rows are exact, with the exception
of the underlined nodes, with the proviso that the bent $\delta_{EL}$
arrow makes the following sequence exact also except at the underlined
nodes:
\begin{equation}\label{eq:bi-exact}\xymatrix{
	0 \ar[r] &
	\underline{\Omega}^{0,0} \ar[r]^-\dh & \cdots
	\underline{\Omega}^{n-1,0} \ar[r]^-\dh &
	\underline{\Omega}^{n,0} \ar[r]^-{\delta_{EL}} &
	\F^1 \ar[r]^-\Dv & \F^2 \ar[r]^-\Dv \ar[r] & \cdots .
}\end{equation}
At the underlined nodes, the horizontal cohomology is completely
characterized by the de~Rham cohomology of $M$, $H^{k,0}(\dh) \cong
H_\dR^k(M)$. At the same nodes, the vertical cohomology consists of all
field-independent forms, $H^{k,0}(\dv) \cong \Omega^k(M)\sso
\Omega^{k,0}(F)$.

Note that any local horizontal form $\omega\in \Omega^{*,0}(F)$
naturally defines a smooth bundle morphism $\omega\colon J^\oo F\to
\Lambda^*M$ as well as a differential operator $\omega\colon \Secs(F)
\to \Secs(\Lambda^*M)$, $\omega[\psi] = \omega\circ j^\oo\psi$ for
$\psi\in\Secs(F)$, where we have slightly abused notation by denoting
each of these naturally associated objects by the same symbol $\omega$.
It is convenient to generalize this construction by replacing the bundle
$\Lambda^*M\to M$ with an arbitrary vector bundle $H\to M$. Any
differential operator $f\colon J^\oo F \to H$ shall also be termed a
\emph{local section of $H$} (also \emph{$F$-local} if such precision is
necessary). Denote the space of all such local sections by
$\Secs_F(H)$. Evidently the space of local horizontal forms is the
same as the space of local sections of $\Lambda^*M$, $\Omega^{*,0}(F)
\cong \Secs_F(\Lambda^*M)$. The pullback along the bundle projection
$F\to M$ induces a natural inclusion $\Secs(H)\to \Secs_F(H)$ of
\emph{field-independent} local sections in the space of all local
sections. Similarly, given two vector bundles $G\to M$ and $H\to M$, we
call a differential operator $f\colon \Secs_F(G) \to \Secs_F(H)$
\emph{local} (or \emph{$F$-local}) when $f[\psi,\xi[\psi]] = f\circ(
j^\oo\psi\times p_\oo\xi)$, where $\psi\in\Secs(F)$, $\xi\in \Secs_F(G)$
and on the right-hand side $f\colon J^\oo(F\times G)\to H$. If
$f[\psi,\zeta]$ is linear in its second argument, we write it as
$f[\psi;\zeta]$. Of course, any differential operator $f\colon\Secs(G)
\to \Secs(H)$ naturally defines a \emph{field-independent} local
differential operator $f\colon \Secs_F(G) \to \Secs_F(H)$.

Local sections and local differential operators are discussed
in~\cite{verbov1,verbov2} as $\mathcal{C}$-modules and
$\mathcal{C}$-differential operators.

\subsection{PDE submanifold}\label{sec:pde}
In the jet bundle setting, a PDE system has a dual
description~\cite{vinogradov,verbov1,verbov2,seiler}, as a
non-empty, smooth sub-bundle $\E^\oo\sso J^\oo(F)$ over $M$ and as the
zero-set of the $\oo$-prolongation $p_\oo e\colon J^\oo F\to J^\oo G$ of
a smooth bundle morphism $e\colon J^\oo F\to G$, where the vector
bundles $F\to M$ and $G\to M$ are respectively referred to as the
\emph{field bundle} and the \emph{equation bundle}. The requirement of
having both descriptions is non-vacuous. It is possible that the zero
set of $p_\oo e$ is not a submanifold (say if the Jacobian of $e$ is not
of constant rank) and it is possible that $\E^\oo$ is not the zero set
of any smooth bundle map (see~\cite[Section 7]{goldschmidt} for the necessary
and sufficient topological condition on the normal bundle of the
embedding $\E^\oo\sso J^\oo(F)$). A section $\phi\colon M\to F$ is a
\emph{solution} of the PDE system if its $\oo$-jet prolongation is
contained in the PDE submanifold, $j^\oo \phi(M) \sse \E^\oo$, or
equivalently if $e[\phi] = 0$.

Clearly the smooth PDE sub-bundle $J^\oo F\supset \E^\oo\to M$ provides
an intrinsic description of the PDE system, while the choice of the
\emph{equation form}, the bundle $G\to M$ and the map $e$, are not
unique. Suppose that the differential operator $e$ is of order $k$ and
denote the natural projection $\pi^k\colon J^\oo F\to J^kF$. Without
loss of generality, we can presume that the zero set $\E$ of $e\colon
J^kF \to G$ is a smooth sub-bundle of $J^kF$ and that it agrees with the
projection of $\E^\oo$, $\E = \pi^k \E^\oo$. Such an equation form is
said to be \emph{involutive}. In the sequel we freely use an equation
form $e[\psi] = 0$ or the submanifold $\E\sso J^kF$ to define a PDE
system.

Define the space of \emph{on-shell local variational forms}
$\Omega^*(\E^\oo)$ to be the pullback image of $\Omega^*(J^\oo F)$ along
the inclusion $\E^\oo\sso J^\oo F$ and $\d$ the de~Rham differential on
it. The horizontal-vertical bigrading and the split of the de~Rham
differential, $\d = \dh+\dv$, commute with the pullback, which
immediately defines the decomposition $\Omega^*(\E^\oo) \cong \sum_{h,v}
\Omega^{h,v}(\E)$ and the \emph{on-shell} variational bicomplex with the
projected operators $\dh$, $\dv$, $\delta_{EL}$ and $\delta_v$. The
cohomology $H_\E^{*,0}(\dh)$ in the space of \emph{on-shell local
horizontal forms} is also called the \emph{characteristic cohomology} of
the PDE system and is denoted
$H^*_\Char(\E)$~\cite{olver,vinogradov,bbh-rep}. We still have a natural
inclusion of the de~Rham cohomology of $M$ in the characteristic
cohomology, $H^*_\dR(M)\sse H^*_\Char(\E)$, but it is no longer
necessarily an isomorphism. The image of this inclusion is called the
subspace of \emph{field-independent} local horizontal forms. The
quotient $H^p_\cur(\E) := H^{n-p}_\Char(\E) / H^{n-p}_\dR(M)$ is known
as the space of \emph{conserved (local) $p$-currents}.

We also define the space of \emph{null local variational forms}
$\hOmega^*(\E^\oo)$ as the kernel of the projection $\Omega^*(F) \to
\Omega^*(\E^\oo)$, which consists of all local variational forms that
vanish on the PDE submanifold $\E^\oo$. Therefore, we have the following
exact sequence:
\begin{equation}\xymatrix{
	\hOmega^{*,*}(\E) \ar[r] &
	\Omega^{*,*}(F) \ar[r] &
	\Omega^{*,*}(\E) .
}\end{equation}
We denote the cohomology of the complex $(\hOmega^{*,0},\dh)$ of null
local horizontal forms by $H^*_\Null(\E)$. It should be noted that,
provided the regularity conditions to be specified below are satisfied,
null local forms are precisely those that are exact in terms of the
Koszul-Tate part of the BV-BRST complex~\cite{bbh-rep}.

Given a vector bundle $H\to M$, we can analogously define the outer ends
of the exact sequence
\begin{equation}\xymatrix{
	\hSecs_\E(H) \to \Secs_F(H) \to \Secs_\E(H) ,
}\end{equation}
where the space $\hSecs_\E(H)$ of \emph{null local sections of $H$}
consists of those that vanish on the PDE submanifold $\E^\oo$ and the
space $\Secs_\E(H)$ of \emph{on-shell local sections of $H$} is the
quotient.

Now, let us consider the equation form $e[\phi] = 0$ as defining a local
differential operator $e\colon \Secs_F(F) \to \Secs_F(G)$ by the formula
$e[\psi,\phi[\psi]] = e[\psi]$, which corresponds to the smooth bundle
map $\id\times p_\oo e\colon J^\oo(F\times F)\to J^\oo G$. Consider the
image $\G^\oo = (\id \times p_\oo e)(J^\oo (F\times F)) \sse J^\oo
(F\times G)$. In this paper we are concerned with the case when $\G^\oo$
is not necessarily all of $J^\oo (F\times G)$. In particular, we are
interested in the cases when $\G^\oo$ is itself a PDE in the sense given
above. That is, there exists an involutive equation form $z^0\colon
J^\oo (F\times G)\to Z^1$ such that $\G^\oo$ is the zero set of
$\pi_F\times p_\oo z^0$ and a submanifold of $J^\oo (F\times G)$. The
naturally associated local differential operator $z^0\colon \Secs_F(G)
\to \Secs_F(Z^1)$ is variously known as a \emph{(stage-0) Noether
operator}, \emph{compatibility operator} or \emph{redundancy operator},
while we call $Z^1\to M$ the \emph{(stage-1) Noether bundle}. When there
are no topological obstructions, the same construction can be iterated.
Let $Z^0=G$, $(\Z^0)^\oo=\G^\oo$ and define $(\Z^i)^\oo = (\pi_F\times
p_\oo z^{i-1})(J^\oo (F\times Z^{i-1}))$ with involutive equation form
$z^i \colon J^\oo (F\times Z^i)\to Z^{i+1}$, with $z^i$ and $Z^i$
respectively the \emph{stage-$i$ Noether operator} and \emph{stage-$i$
Noether bundle} (which also naturally define local differential
operators). Provided there is no topological obstruction and the iteration
terminates%
	\footnote{Actually, at each stage, this procedure need only work on an
	open neighborhood of $(\Z^r)^\oo$, though then $Z^{r+1}$ may need to
	be chosen as a non-linear smooth bundle. For simplicity, we ignore
	these technicalities and work only with vector bundles over $M$.} %
at $i=r$ if $(\Z^r)^\oo=J^\oo Z^r$, the PDE system in question is said
to be \emph{stage-$r$ irreducible}. When $r>0$, the PDE system is called
a \emph{gauge theory}, while when $r=1$ it is said to be an
\emph{irreducible gauge theory}. The end point of this construction is a
\emph{formally exact} complex of local differential operators
\begin{equation}\label{eq:Ncpx}\xymatrix{
	\Secs_F(F) \ar[r]^-{e} &
	\Secs_F(Z^0) \ar[r]^-{z^0} &
	\Secs_F(Z^1) \ar[r]^-{z^1} &
	\cdots \Secs_F(Z^r) \ar[r] & 0 ,
}\end{equation}
where \emph{formal exactness}~\cite{verbov1,verbov2} means that the
following is an exact sequence of smooth bundles over $M$:
\begin{multline}\xymatrix{
	J^\oo F^2 \ar[rr]^-{\id\times p_\oo e} &&
	J^\oo (F\times Z^0) \ar[rr]^-{p_\oo (\pi_F\times z^0)} &&
	J^\oo (F\times Z^1)
} \\
\xymatrix{
	\ar[rr]^-{p_\oo (\pi_F\times z^1)} &&
	\cdots J^\oo (F\times Z^r) \ar[r] & M\times\{*\} .
}\end{multline}

To even have a hope of proving some stability results for the space of
solutions of a PDE system, the PDE system itself must satisfy some
regularity properties. Moreover, in the next section, we state a
characterization of the characteristic cohomology groups in terms of the
Noether complex, which only holds when appropriate regularity properties
are satisfied. These properties are discussed in detail in~\cite[Sections
5.1, 6.4.2--3]{bbh-rep} and are summarized in the following subsections.

\subsubsection{Local regularity}
We say that a PDE system $\E^\oo\sso J^\oo F$ is \emph{locally regular}
if it is (i) \emph{stage-$r$ irreducible} and (ii) the corresponding
differential operators $e\colon J^k\colon F\to G$, $z^i\colon J^{k_i}
(F\times Z^i) \to Z^{i+1}$, where $e$ and $z^i$ are respectively of
orders $k$ and $k_i$, can be chosen to be smooth bundle maps with
Jacobians of constant rank and such that the Noether operators
$z^i[\psi,\zeta]$ are linear in their second arguments, which we denote
as $z^i[\psi;\zeta]$.

This requirement is essentially equivalent to that of~\cite[Section
5.1]{bbh-rep}. In particular, it allows us to conclude that any smooth
bundle map $f\colon J^\oo F\to F'$ that vanishes on the PDE manifold
$\E^\oo$ (i.e.,\ $e[\psi]=0$ implies $f[\psi]=0$) must factor through the
prolongued equation form:
\begin{equation}\xymatrix{
	f\colon J^\oo F \ar[r]^-{p_\oo e} &
	J^\oo (F\times G) \ar[r]^-{f} &
	F' ,
}\end{equation}
that is, $f[\psi] = f[\psi;e[\psi]]$, where $f[\psi;\xi]$ is linear in
the last argument. Further more, we can also conclude that any smooth
bundle map $g\colon J^\oo (F\times G)\to G'$ that vanishes on the stage-$0$
Noether manifold $\G^\oo = (\Z^0)^\oo$ (i.e.,\ $g[\psi,e[\psi]]=0$ for
arbitrary $\psi$) must factor through the prolongued stage-$0$ Noether
operator:
\begin{equation}\xymatrix{
	g\colon J^\oo (F\times G) \ar[r]^-{p_\oo e} &
	J^\oo (F\times Z^1) \ar[r]^-{g} &
	G' ,
}\end{equation}
that is, $g[\psi,\xi] = g[\psi; z^0[\psi;\xi]]$, where $g[\psi;\zeta]$ is
linear in the last argument. Similar remarks can be made about the
higher stage Noether operators. In total, they are equivalent to the
acyclicity of the Koszul-Tate complex in positive anti-field degree,
which is used extensively in the BV-BRST analysis of gauge
theories~\cite{bbh-rep}.

\subsubsection{Local linearizability}
In the question of linearization stability, we are interested in
the relationship between the space of linearized solutions about a
background exact solution $\varphi\colon M\to F$ and a neighborhood of
$\varphi$ in the space of exact solutions. As such, the linearization of
the PDE system in question must be well defined.

Consider an arbitrary smooth $1$-parameter family of sections
$\phi_t\colon M\to F$, with $\phi_t = \varphi + t\psi + \Psi_t$, where
$\Psi_t = O(t^2)$. For convenience, we denote $\psi_t = t\psi + \Psi_t$.
We call a stage-$r$ irreducible, locally regular PDE system
\emph{locally linearizable about $\varphi$} if (i) we can write, for
$\phi_t$ as above and any $\zeta^i\in \Secs(Z^i)$,
\begin{align}
\label{eq:lin-e}
	e[\phi_t] &= e[\varphi] + e_\varphi[\psi_t] - f_\varphi[\psi_t] , \\
\label{eq:lin-z}
	z^i[\phi_t;\zeta^i] &= z^i_\varphi[\zeta^i] - y^i_\varphi[\psi_t;\zeta^i] ,
\end{align}
where the local differential operators $e_\varphi$ and $z^i_\varphi$
($z^i_\varphi[\zeta^i] = z^i[\varphi;\zeta^i]$) are linear, while
$f_\varphi[\psi_t]=O(t^2)$ and $y_\varphi[\psi_t;\zeta] = O(t)$, and
(ii) the linear PDE system defined by $e_\varphi[\psi] = 0$ is stage-$r$
irreducible and locally regular, with its \emph{Noether complex}%
	\footnote{Also known as the \emph{compatibility complex} or
	\emph{Janet sequence}, cf.~footnote~\ref{noether-term}.} %
given by
\begin{equation}\label{eq:lNcpx}\xymatrix{
	\Secs_F(F) \ar[r]^-{e_\varphi} &
	\Secs_F(Z^0) \ar[r]^-{z^0_\varphi} &
	\Secs_F(Z^1) \ar[r]^-{z^1_\varphi} &
	\cdots \Secs_F(Z^r) \ar[r] & 0 .
}\end{equation}
We denote the PDE submanifolds defined by the linearized equation form
$e_\varphi[\psi] = 0$ by $\E_\varphi^\oo\sse J^\oo F$ and
$\E_\varphi = \pi^{k_\varphi} \E_\varphi^\oo \sso J^{k_\varphi}F$, where
$k_\varphi$ is the order of $e_\varphi$. Incidentally, linearization has
made the local differential operators $e_\varphi$ and $z^i_\varphi$
field-independent. We can then naturally define the
\emph{field-independent linearized Noether complex}
\begin{equation}\label{eq:ilNcpx}\xymatrix{
	\Secs(F) \ar[r]^-{e_\varphi} &
	\Secs(Z^0) \ar[r]^-{z^0_\varphi} &
	\Secs(Z^1) \ar[r]^-{z^1_\varphi} &
	\cdots \Secs(Z^r) \ar[r] & 0 ,
}\end{equation}
which is a subcomplex of the field-dependent one above.

The requirement of local linearizability also imposes a restriction on
the allowed background sections $\varphi\colon M\to F$. We call such an
admissible section $\varphi$ a \emph{linearizable background}. Denote by
$\Secs_\lin(F)\sse \Secs(F)$ the subset of linearizable background
sections. Also, denote by $\S(\E) \sso \Secs(F)$ the subset of
solutions, $\varphi\in \S(\E)$ if $e[\varphi] = 0$. Finally, denote by
$\S_\lin(\E)\sse \S(\E)$ the subset of solutions that consists of all
linearizable backgrounds. All of these spaces can be topologized as
subsets of $\Secs(F)$, with its natural smooth compact open (or Whitney)
Fr\'echet topology~\cite{hirsch,km}. We say that the PDE system in question is
\emph{locally linearizable} if the $\S_\lin(\E)$ is an open subset of
$\S(\E)$.

\subsubsection{Local normality}
Finally, there is almost no hope of identifying linearized solutions
with the full set of exact solutions if the non-linear system is of a
higher order than the linear one.

Let $k$, $k_i$, $k^\varphi$ and $k_i^\varphi$ denote the respective
orders of the differential operators $e$, $z^i$, $e_\varphi$ and
$z^i_\varphi$. We say that a stage-$r$ irreducible, locally
regular, locally linearizable PDE system is \emph{locally normal} if (i)
the orders $k=k^\varphi$, $k_i = k_i^\varphi$ agree for all linearizable
backgrounds $\varphi\in \S_\varphi(\E)$ and (ii) the ranks of the linear
bundle maps $e_\varphi \colon J^k F\to G$ and $z^i_\varphi
\colon J^{k_i}(F\times Z^i)\to Z^{i+1}$ are respectively the same as
those of the Jacobians of the smooth bundle maps $e\colon J^kF\to G$ and
$z^i\colon J^{k_i}(F\times Z^i) \to Z^{i+1}$.

\section{Deformations, cosymmetries, sources}\label{sec:def-cosym-src}

\subsection{Consistent deformations}\label{sec:def}
Consider the linearized Noether complex~\ref{eq:lNcpx}. Clearly, by
linearity, the subspaces of null local sections, $\hSecs_{\E_\varphi}(F)\sso
\Secs_F(F)$ and $\hSecs_{\E_\varphi} (Z^i)\sso \Secs_F(Z^i)$ are
preserved by the action of $e_\varphi$ and $z^i_\varphi$. Therefore, we
can define the \emph{null} and \emph{on-shell} Noether complexes as the
top and bottom rows of the following commuting bicomplex:
\begin{equation}\vcenter{\xymatrix{
	\hSecs_{\E_\varphi}(F) \ar[r]^-{e_\varphi} \ar[d] &
	\hSecs_{\E_\varphi}(Z^0) \ar[r]^-{z^0_\varphi} \ar[d] &
	\hSecs_{\E_\varphi}(Z^1) \ar[r]^-{z^1_\varphi} \ar[d] &
	\cdots \Secs_F(Z^r) \ar[r] \ar[d] & 0 \\
	\Secs_F(F) \ar[r]^-{e_\varphi} \ar[d] &
	\Secs_F(Z^0) \ar[r]^-{z^0_\varphi} \ar[d] &
	\Secs_F(Z^1) \ar[r]^-{z^1_\varphi} \ar[d] &
	\cdots \Secs_F(Z^r) \ar[r] \ar[d] & 0 , \\
	\Secs_{\E_\varphi}(F) \ar[r]^-{e_\varphi} &
	\Secs_{\E_\varphi}(Z^0) \ar[r]^-{z^0_\varphi} &
	\Secs_{\E_\varphi}(Z^1) \ar[r]^-{z^1_\varphi} &
	\cdots \Secs_F(Z^r) \ar[r] & 0
}}\end{equation}
where the columns are exact and the middle row is formally exact. We
call the cohomologies of the rows of the above bicomplex, respectively,
the \emph{null}, \emph{(off-shell)} and \emph{on-shell (stage-$i$)
consistent deformations}, denoted respectively by
$\hat{H}^i_\Def(e_\varphi)$, $H^i_\Def(F,e_\varphi)$ and
$H^i_\Def(e_\varphi)$. Note that, since the Noether complex depends
explicitly on the equation form $e_\varphi$, rather than just on the PDE
submanifold $\E_\varphi$, we indicate the same dependence in the
notation for consistent deformations.

From the definition, an on-shell consistent deformation class $[f] \in
H^0_\Def(e_\varphi)$ is represented by a local section $f\in
\Secs_F(Z_0)=\Secs(G)$. The local section $f$ is said to be
\emph{trivial} if $[f] = [0]$, which by local regularity means that it
must be of the form
\begin{equation}\label{eq:cdr-triv}
	f[\psi] = e_\varphi[g[\psi]] + h[\psi;e_\varphi[\psi]] ,
\end{equation}
for some local section $g\in \Secs_F(F)$ and some local differential
operator $h\colon \Secs_F(G) \to \Secs_F(G)$, with $h[\psi;\zeta]$
linear in its second argument.

Recall the linearization formulas~\ref{eq:lin-e} and~\ref{eq:lin-z}. Let
us define the leading order $m>1$ and leading term $f^{(m)}_\varphi$ of
$f_\varphi$ as
\begin{equation}
	f_\varphi[t\psi] = t^m f^{(m)}_\varphi[\psi] + O(t^{m+1}) .
\end{equation}
Similarly, let $m_i>1$ and $y^{i(m_i)}_\varphi$ the leading order and
leading term of $y^i_\varphi$. We also use the same notation for higher
order terms in $f$ and $y^i$. Clearly, $f^{(m)}_\varphi[\psi]$ and
$y^{i(m_i)}_\varphi[\psi;\zeta]$ are homogeneous in $\psi$, of order $m$
and $m_i$ respectively. Consider the following expansion of the exact
Noether identity $z^0[\phi;e[\phi]]=0$, which holds for arbitrary
sections $\phi = \varphi + t\psi\in \Secs(F)$:
\begin{multline}
	t z^0_\varphi[e_\varphi[\psi]] - z^0_\varphi[f_\varphi[t\psi]]
		- t y^0_\varphi[t\psi;e_\varphi[\psi]]
		+ y^0_\varphi[t\psi;f_\varphi[t\psi]] = 0 \\
	\implies
	t z^0_\varphi[e_\varphi[\psi]]
	= t^m z^0_\varphi[f^{(m)}_\varphi[\psi]]
		+ t y^0_\varphi[t\psi;e_\varphi[\psi]]
		+ O(t^{m+1}) .
\end{multline}
Completing the Taylor expansion and comparing the coefficients of powers
of $t$, we find the condition $z^0_\varphi[e_\varphi[\psi]] = 0$ at
order $O(t)$, which was already a requirement of local linearizability,
and at order $O(t^m)$ that
\begin{equation}
	z^0_\varphi[f^{(m)}_\varphi[\psi]]
	= - y^{0(m-1)}_\varphi[\psi;e_\varphi[\psi]] .
\end{equation}
In other words, we find that $f^{(m)}_\varphi$ represents an equivalence
class $[f^{(m)}_\varphi]\in H^0_\Def(e_\varphi)$ of on-shell consistent
deformations.

Of course, one could consider the cohomology spaces
$H^i_\Def(M,e_\varphi)$ of the field-independent linearized Noether
complex~\eqref{eq:ilNcpx}. These \emph{field-independent} consistent
deformations are naturally included in those that are possibly field
dependent, $H^i_\Def(M,e_\varphi)\sse H^i_\Def(F,e_\varphi)$. Of course,
any non-trivial field-independent consistent deformation is still
non-trivial on-shell, while any non-trivial null consistent deformation
cannot be field-independent. Since in this paper we are chiefly
concerned with non-linearities, all consistent deformations considered
in the sequel will be field-dependent and of at least quadratic order.

\subsection{Cosymmetries and conserved currents}\label{sec:cosym}
For PDE systems that satisfy the regularity, linearizability and
normality conditions given previously in Section~\ref{sec:pde}, it is
well known that there is a direct correspondence between the spaces
$H^p_\cur(\E_\varphi)$ of conserved local $p$-currents (which is the
same as the space $H^{(n-p)}_\Char(\E)/H^{(n-p)}_\dR(M)$ of on-shell
closed local $(n-p)$-forms modulo field-independent ones) and the spaces
of on-shell stage-$(p-1)$ cohomology of the complex formally adjoint to
the linearized Noether complex~\ref{eq:lNcpx}. We refer to the latter
objects as \emph{stage-$(p-1)$ cosymmetries} and denote their space by
$H^{(p-1)}_\cosym(e_\varphi)$; they are defined more precisely below.
For simplicity, since we need only deal with the linearized Noether
complex below, we restrict our discussion of this correspondence purely
to the case of linear PDE systems.

If $h\colon J^\oo F\to H$ is a linear differential operator, we define
its \emph{formal adjoint} $h^*$ as follows. For any vector bundle $H\to
M$, denote by $H^*\to M$ its dual bundle and by $\tilde{H}^* = \Lambda^n
M\otimes H^*$ the \emph{densitized dual} bundle. There is a natural
fiber-wise pairing between sections $\eta\colon M\to H$ and dual
densities $\tilde{\alpha}^*\colon M\to \tilde{H}^*$, $\eta\cdot
\tilde{\alpha}^* \colon M\to \mathbb{R}$. The \emph{formal adjoint} $h^*
\colon J^\oo\tilde{H}^* \to \tilde{F}^*$ is defined as the unique linear
differential operator that satisfies
\begin{equation}\label{eq:green-form}
	h[\psi]\cdot \tilde{\alpha}^* - \psi\cdot h^*[\tilde{\alpha}^*]
	= \d \G_h[\psi,\tilde{\alpha}^*]
\end{equation}
for arbitrary sections $\psi\colon M\to F$ and $\tilde{\alpha}^*\colon
M\to \tilde{H}^*$ and some bilinear bidifferential operator
$\G_h[\psi,\tilde{\alpha}^*] \colon J^\oo(F\times \tilde{H}^*) \to
\Lambda^{n-1}M$, where $\G_h$ is called a \emph{Green form} associated
to the adjoint pair $h$ and $h^*$. Note that we can consider $\G_h$ as
an $F\times \tilde{H}^*$-local horizontal form, $\G_h\in
\Omega^{n-1,0}(F\times \tilde{H}^*)$.  Also, $\G_h$ is only defined up
to the addition of an exact local horizontal form
$\G_h[\psi,\tilde{\alpha}^*] \sim \G_h[\psi,\tilde{\alpha}^*] + \d_h
\H[\psi,\tilde{\alpha}^*]$, where $\H$ is itself a bilinear
bi-differential operator, so only the corresponding equivalence class
$[\G_h]$ is well defined, though we may restrict our attention only to
representatives that are bilinear in their arguments. If $h[\psi;\xi]$
is a differential operator that is linear only in its second argument,
we can still define its formal adjoint with respect to the second
argument alone. It is again denoted by $h^*[\psi;\tilde{\alpha}^*]$ and
is also linear in its second argument, while an associated Green form is
denoted by $\G_h[\psi;\xi,\tilde{\alpha}^*]$ and is bilinear in its last
two arguments.

It is straight forward to verify that for a stage-$r$ irreducible,
locally regular, locally linearizable, locally normal PDE system, the
formal adjoint of its linearized Noether complex is also formally exact.
Also, since all the operators involved are linear, the subspaces of null
local sections, $\hSecs_{\E_\varphi}(\tilde{F}^*)\sso
\Secs_F(\tilde{F}^*)$ and $\hSecs_{\E_\varphi}(\tilde{Z}^{r*}) \sso
\Secs_F(\tilde{Z}^{r*})$, are preserved. That is, we can construct the
following commuting bicomplex, where the columns are exact, while the
middle row is formally exact:
\begin{equation}\label{eq:alNcpx}\vcenter{\xymatrix{
	\hSecs_{\E_\varphi}(\tilde{F}^*) \ar@{<-}[r]^-{e_\varphi^*} \ar[d] &
	\hSecs_{\E_\varphi}(\tilde{Z}^{0*}) \ar@{<-}[r]^-{z^{0*}_\varphi} \ar[d] &
	\hSecs_{\E_\varphi}(\tilde{Z}^{1*}) \ar@{<-}[r]^-{z^{1*}_\varphi}\ar[d]&\cdots
	\hSecs_{\E_\varphi}(\tilde{Z}^{r*}) \ar@{<-}[r] \ar[d] & 0
	\\
	\Secs_F(\tilde{F}^*) \ar@{<-}[r]^-{e_\varphi^*} \ar[d] &
	\Secs_F(\tilde{Z}^{0*}) \ar@{<-}[r]^-{z^{0*}_\varphi} \ar[d] &
	\Secs_F(\tilde{Z}^{1*}) \ar@{<-}[r]^-{z^{1*}_\varphi} \ar[d]& \cdots
	\Secs_F(\tilde{Z}^{r*}) \ar@{<-}[r] \ar[d] & 0 .
	\\
	\Secs_{\E_\varphi}(\tilde{F}^*) \ar@{<-}[r]^-{e_\varphi^*} &
	\Secs_{\E_\varphi}(\tilde{Z}^{0*}) \ar@{<-}[r]^-{z^{0*}_\varphi} &
	\Secs_{\E_\varphi}(\tilde{Z}^{1*}) \ar@{<-}[r]^-{z^{1*}_\varphi} & \cdots
	\Secs_{\E_\varphi}(\tilde{Z}^{r*}) \ar@{<-}[r] & 0
}}\end{equation}
Finally, we define (for $r\ge 0$) \emph{null}, \emph{(off-shell)} and
\emph{on-shell local (stage-$r$) cosymmetries} to be elements of the
kernel of the local differential operator $z^{i*}_\varphi$ on,
respectively, the top, middle and bottom rows of the above bi-complex;
by convention we define $z^{-1}_\varphi = e_\varphi$ and $Z^{-1} = F$.
That is, $\tsxi{r}\in \Secs_F(\tilde{Z}^{r*})$ represents an on-shell
local cosymmetry if $z^{(r-1)*}_\varphi[\psi;\tsxi{r}[\psi]] = 0$
whenever $e_\varphi[\psi] = 0$. The cohomologies of these rows define
equivalence classes of cosymmetries. A cosymmetry $\tsxi{r}\in
\Secs_F(\tilde{Z}^{r*})$ is \emph{trivial} if, respectively,
\begin{align}
	\text{(null)} \quad
	\tsxi{r}[\psi] &= z^r_\varphi[\tsxi{r+1}[\psi;e_\varphi[\psi]]] , \\
	\text{(off-shell)} \quad
	\tsxi{r}[\psi] &= z^r_\varphi[\tsxi{r+1}[\psi]] , \\
	\text{(on-shell)} \quad
	\tsxi{r}[\psi] &= z^r_\varphi[\tsxi{r+1}[\psi]] +
		\tilde{\zeta}^{r*}[\psi;e_\varphi[\psi]] .
\end{align}
Two representatives are in the same equivalence class precisely when
they differ by a trivial one. The spaces of equivalence classes of null,
off-shell and on-shell stage-$r$ local cosymmetries, respectively, by
$\hat{H}^r_\cosym(e_\varphi)$, $H^r_\cosym(F,e_\varphi)$ and
$H^r_\cosym(e_\varphi)$. We call these equivalence classes
\emph{rigid cosymmetries}. The definition of a cosymmetry explicitly
involves the equation form $e_\varphi$, rather than just the intrinsic
PDE submanifold $\E_\varphi$, and the notation reflects that.

Since all operators involved in the definition of the adjoint linearized
Noether complex are field-independent, we can also define the
\emph{field-independent} adjoint linearized Noether complex
\begin{equation}\label{eq:ialNcpx}\xymatrix{
	\Secs(\tilde{F}^*) \ar@{<-}[r]^-{e_\varphi^*} &
	\Secs(\tilde{Z}^{0*}) \ar@{<-}[r]^-{z^{0*}_\varphi} &
	\Secs(\tilde{Z}^{1*}) \ar@{<-}[r]^-{z^{1*}_\varphi} & \cdots
	\Secs(\tilde{Z}^{r*}) \ar@{<-}[r] & 0 ,
}\end{equation}
whose cohomology spaces we denote by $H^r_\cosym(M,e_\varphi)$. We have
the natural inclusion $H^r_\cosym(M,e_\varphi) \sse
H^r_\cosym(F,e_\varphi)$. Of course, any non-trivial field-independent
local cosymmetry is still non-trivial on-shell, while any non-trivial
null local cosymmetry cannot be field-independent.

We conclude this section recalling the following known bijection between
equivalence classes of higher stage cosymmetries and higher conserved
currents. We state the result only for the linearized equation, though a
similar result works for non-linear equations as well.
\begin{prop}[Generalized Noether's first theorem]\label{prp:gN1}
For the linearized PDE system $\E_\varphi$, the space of classes of
conserved local currents is isomorphic to the space of classes of null
local forms, which in turn is isomorphic to the space of classes of
rigid on-shell local cosymmetries. That is, for $0 < p \le n$,
\begin{equation}
	H^{p+1}_\cur(\E_\varphi)
	\cong H^{n-p}_\Null(\E_\varphi)
	\cong H^p_\cosym(e_\varphi) .
\end{equation}
\end{prop}
Though the proof of Proposition~\ref{prp:gN1} can be found in several
places in the literature~\cite{bbh-rep,verbov1,verbov2}.

The name \emph{cosymmetry} is meant to be evocative. For a variational
PDE system (one defined by the Euler-Lagrange equations of a local
Lagrangian), the linearized equation form $e_\varphi\colon \Secs_F(F)
\to \Secs_F(\tilde{F}^*)$ is formally self-adjoint, $e_\varphi^* =
e_\varphi$.  In that case, the adjoint linearized Noether complex is
identical to the formally exact complex of generators of \emph{(higher
stage) gauge transformations} (generalized Noether's second theorem). It
is well known that the cohomology of higher stage gauge generators is
identified with higher stage rigid symmetries, those of stage-$0$ being
the ordinary symmetries. Stage-$1$ rigid symmetries have also been
called \emph{reducibility parameters}~\cite{barnich-brandt}. For
non-variational system, there need not be any correspondence between
gauge generators and Noether operators or between symmetries and
cosymmetries, hence the need for the distinct terminology.

In the case of stage-$0$ irreducible relativistic field theories and the
Einstein or Yang-Mills gauge theories, rigid symmetries are known to
correspond to so-called Killing vectors (or generalizations
thereof~\cite{moncrief-killing1,moncrief-ym}). By analogy, we can call
the condition
\begin{equation}\label{eq:co-kil}
	z^{(r-1)*}_\varphi[\tsxi{r}[\psi]] = 0
	~~\text{and}~~
	\tsxi{r}[\psi] \ne z^{r*}_\varphi[\tsxi{(r+1)}[\psi]
	~~\text{(on-shell)} ,
\end{equation}
which defines a non-trivial cosymmetry, the \emph{(generalized) co-Killing
condition}.

\subsection{Null sources}\label{sec:src}
A \emph{null local $p$-source} $\rho$ is a local horizontal
$(n-p)$-form, $\rho\in \Omega^{n-p,0}(F)$, such that (i) it is
horizontally exact, $\rho = \dh j$ for some local horizontal
$(n-p-1)$-form $j\in \Omega^{n-p-1,0}(F)$, and (ii) it vanishes on
linearized solutions, that is, it pulls back to $0$ in $\Omega^{n-p,0}(\E_\varphi)$ along
the inclusion $\E_\varphi^\oo \sso J^\oo F$, or simply $\rho\in
\hOmega^{n-p,0}(\E_\varphi)$. A null local source $\rho$ is said to be
\emph{trivial} if $\rho = \dh j$ where $j$ itself vanishes on solutions.
Two null local sources are said to be \emph{equivalent} if they differ
by a trivial one. We denote the space of null local $p$-source classes
by $H^p_\src(\E_\varphi)$.

The term \emph{source} is meant to be evocative. Consider an equation of
the form
\begin{equation}
	\d j = \rho ;
\end{equation}
if the left-hand side can be considered as the divergence of a current,
then the right-hand side should be considered the source (or source
density) for that current, whence the exactness requirement.

Null sources are clearly related to null local forms, as well as to
characteristic cohomology.
\begin{lem}\label{lem:null-src}
We have $H^n_\src(\E_\varphi) = H^0_\Null(\E_\varphi) = 0$ and for $0
\le p < n$
\begin{equation}
	H^p_\src(\E_\varphi) \cong
	H^{n-p}_\Null(\E_\varphi) \cong
	H^{n-p-1}_\Char(\E_\varphi)/H^{(n-p-1)}_\dR(M) \cong
	H^{p+1}_\cur(\E_\varphi) .
\end{equation}
\end{lem}
\begin{proof}
If follows directly from the above definition that for each null local
$p$-source $\rho$, we can find a local horizontal $(n-p-1)$-form $j$
such that $\dh j = \rho$. Clearly $j$ cannot be field-independent,
unless it vanishes, and $\dh j$ vanishes on-shell.  But this means
precisely that $j$ is a conserved $(p+1)$-current. Conversely, for each
conserved $(p+1)$-current $j$, we can define $\rho = \dh j$ and easily
check that $\rho$ is a null $p$-source. Moreover, this correspondence
respects equivalence classes. In other words, the classes of null
$p$-sources are in bijection with the equivalence classes of conserved
$(p+1)$-currents. The isomorphisms with $H^{n-p}_\Null(\E_\varphi)$ and
$H^{n-p-1}_\Char(\E_\varphi)/H^{n-p-1}_\dR(M)$ follow directly from
definitions.

The exactness of the variational bicomplex shows that any local
horizontal $0$-form is horizontally closed only if it is
field-independent and constant. Therefore, any horizontally closed, null
horizontal $0$-form must be trivial. On the other hand, any null
$n$-source must also be trivial, since the only exact local horizontal
form in this degree is zero. In other words, $H^n_\src(\E_\varphi) =
H^0_\Null(\E_\varphi) = 0$.
\end{proof}

Here is another convenient way to represent null sources.
A local horizontal $(n-p)$-form $\rho$ naturally defines a local
differential operator $\rho\colon \Secs(M)\to \Omega^{n-p}(M)$. If
$\rho$ is a null local source, local regularity of the linearized PDE
system implies that it can be written as $\rho = \rho(e_\varphi)$, where
on the right-hand side we denote the local section $e_\varphi\in
\Secs_F(G)$ is acted on by the linear local differential operator
$\rho\colon \Secs_F(G) \to \Omega^{n-p,0}(F)$. We write this as
$\rho[\psi] = \rho[\psi; e_\varphi[\psi]]$. The local differential
operator $\rho[\psi; \zeta]$ is not unique, since any other one of the
form $\rho[\psi; \zeta] + \sigma[\psi; z^0[\zeta]] +
\tau[\psi,e_\varphi[\psi],\zeta]$, where $\sigma$ is linear in its
second argument and $\tau$ is bilinear and anti-symmetric in its last
two arguments, would represent the same null $p$-source $\rho[\psi]$.
Local regularity implies that these possibilities exhaust the available
ambiguity. We call a linear local differential operator $\rho\colon
\Secs_F(G)\to \Omega^{*,0}(F)$ \emph{trivial} if it is of the form
\begin{equation}\label{eq:nsr-triv}
	\rho[\psi;\zeta] = \d j[\psi;\zeta] + \sigma[\psi;z^0_\varphi[\zeta]]
		+ \tau[\psi;e_\varphi[\psi],\zeta] .
\end{equation}
We have just observed that the space of equivalence classes of null
$p$-sources $H^p_\src(\E_\varphi)$ is isomorphic to the space of
equivalence classes $[\rho]$ of linear local differential operators
$\rho\colon \Secs_F(G) \to \Omega^{n-p,0}(F)$ as above, modulo trivial
ones.

\section{Linearization instability}\label{sec:lin-stab}
Consider a PDE system that is stage-$r$ irreducible, locally regular,
locally linearizable and locally normal, with the defining equation form
and Noether complex~\eqref{eq:Ncpx}. Similarly, its linearized equation
form and Noether complex about a linearizable background solution
$\varphi\in\S_\lin(\E)$ are given by~\ref{eq:lNcpx}. In this section, we
shall refer to the space $\S_\lin(\E)$ of linearizable background
solutions simply as the \emph{(exact) solution space}. We shall refer to
the space $\S(\E_\varphi)$ of solutions to the linearized PDE system at
$\varphi$ as the \emph{linearized solution space} at $\varphi$.
Both the exact and linearized solution spaces can be endowed with a
topology as subspaces of the total space of sections $\Secs(F)$. The
choice of topology on $\Secs(F)$ should be adapted to the problem at
hand. The work on linearization stability that followed in the footsteps
of~\cite{fima-linstab} identifies the space of solutions with the space
of Cauchy data satisfying initial value constraints, which is endowed
with a norm topology. Thus, the accompanying functional analytical steps
that require completeness introduce some weak solutions to the elliptic
initial value constraints along with smooth classical ones. On the other
hand, restricting our attention to a single Cauchy surface, rather than
an open spacetime domain, does not take into account the possibility of
singularity formation withing the domain, which excludes some initial
data from the solution space. The seminal result on the global
non-linear stability of Minkowski space~\cite{chr-kl} can be seen through the
prism of linearization stability analysis. In that work, another normed
topology was used on $\Secs(F)$. It would be interesting to also
examine the same questions using a more natural Fr\'echet (Whitney)
topology~\cite{hirsch,km} in which $\Secs(F)$ is complete, without the
need to introduce weak solutions, beyond the minimal attention this
question has already received in the existing literature.

We say that the background solution $\varphi\in \S_\lin(\E)$
is \emph{linearization stable} if (i) there is a neighborhood $U\sso
\S_\lin(\E)$ of $\varphi$ that is homeomorphic to the linearized
solution space $\S(\E_\varphi)$ and (ii) each $1$-parameter family
of linearized solutions of the form $t\psi$, $\psi\in \S(\E_\varphi)$ is
mapped by this homeomorphism to a smooth $1$-parameter family of exact
solutions, $t\psi \mapsto \phi_t$, where $\phi_t\in \S_\lin(\E)$ and
$\phi_t = \varphi + t\psi + O(t^2)$.

Obviously proving linearization stability is a difficult problem that
must involve a significant amount of functional analysis and at present
it has only been treated in detail for a few selected equations. On the
other hand, it is much easier to show that linearization stability is
obstructed. We devote the following subsections to explicitly exhibiting
global geometric conditions on the manifold $M$ and a background
solution $\varphi$ on it, which we call \emph{co-Killing conditions},
that imply the existence of obstructions, also called
\emph{linearization instabilities}.

\subsection{Linearization obstructions}\label{sec:obstr}
Linearization stability is \emph{obstructed} at a linearized solution
$\psi\in\S(\E_\varphi)$ if it cannot be extended to a smooth
$1$-parameter family of exact solutions of the form $\phi_t = \varphi +
t\psi + O(t^2)$; if such a $1$-parameter family does exist, then $\psi$
is called \emph{extendable}. A \emph{stability obstruction} is a
function on the linearized solution space, $Q\colon \S(\E_\varphi)\to
V$, where $V$ is some vector space, such that $Q(\psi) = 0$ whenever
$\psi$ is extendable. A linearization obstruction is said to be
\emph{trivial} if $Q(\psi) = 0$ for every linearized solution. The
obstruction is said to be of \emph{order $m$} if it is homogeneous of
the same order, $Q(t\psi) = t^m Q(\psi)$.

\subsection{Deformation currents from null sources}\label{sec:def-cur}
We are now ready to prove a theorem that relates null sources, consistent
deformations and conserved currents. An easy consequence of it will be
the existence of linearization obstructions, to be discussed in the next
section.
\begin{thm}[Deformation currents]\label{thm:def-cur}
Provided a linear PDE system $\E_\varphi$ with equation form
$e_\varphi[\psi]=0$ is locally regular, there exists a natural
bilinear mapping pairing a null local $p$-source class
with a stage-$0$ consistent deformation class giving a conserved local
$p$-current.
\begin{equation}
	j\colon H^p_\src(\E_\varphi) \times H^0_\Def(e_\varphi) \to
	H^p_\cur(\E_\varphi) ,
	\quad
	([\rho],[f]) \mapsto [j_{\rho,f}] .
\end{equation}
We call $j_{\rho,f}$ the \emph{deformation current} associated to $\rho$
and $f$.
\end{thm}
\begin{proof}
As discussed in the last two sections, provided the regularity condition
is satisfied, each null local $p$-source class $[\rho]$ can be
represented by a linear local differential operator $\rho\colon
\Secs_F(G) \to \Omega^{n-p,0}(F)$, while a stage-$0$ on-shell consistent
deformation class $[f]$ can be represented by a local section $f\in
\Secs_F(G)$. We define the local horizontal $(n-p)$-form $j_{\rho,f} \in
\Omega^{n-p,0}(F)$ by the formula $j_{\rho,f} = \rho(f)$; in other
words,
\begin{equation}
	j_{\rho,f}[\psi] = \rho[\psi;f[\psi]] ,
\end{equation}
for any section $\psi\colon M\to F$. It now remains to check that
$j_{\rho,f}$ is in fact conserved and that the map $j$ is well defined
on equivalence classes.

We have already established that the horizontal differential of a null
$p$-source representative $\rho$ has the form
\begin{equation}
	\d \rho[\psi;\zeta] = \sigma[\psi;z^0_\varphi[\zeta]] +
	\tau[\psi;e_\varphi[\psi],\zeta] .
\end{equation}
The second term already vanishes on-shell because it is linear in the
$e_\varphi[\psi]$ argument. If we set $\zeta=f[\psi]$, the first term on
the right-hand also vanishes on-shell, since $\sigma$ is linear in its
second argument and $z^0_\varphi[f[\psi]]$ vanishes on-shell by the
defining property of a stage-$0$ on-shell consistent deformation.
Therefore, $j_{\rho,f}$ is in fact a conserved local $p$-current.

If $\rho$ is a trivial null local $p$-source, then it must take the form
given in Equation~\eqref{eq:nsr-triv}, provided local regularity holds.
We check that each of the three possible terms is a trivial conserved
$p$-current after setting $\zeta = f[\psi]$. The term $\d j[\psi;\zeta]$
is trivial because it is exact.  The term $\sigma[\psi;
z^0_\varphi[f[\psi]]]$ is trivial because $z^0_\varphi[f[\psi]]$
vanishes on-shell by the definition of a consistent deformation.
Finally, the term $\tau[\psi;e_\varphi[\psi], f[\psi]]$ vanishes
on-shell because it is linear in the $e_\varphi[\psi]$ argument.

If $f$ is a trivial on-shell consistent deformation, then it must take
the form given in Equation~\eqref{eq:cdr-triv}, provided local
regularity holds. We check that each possible term substituted for
$\zeta$ in $\rho[\psi;\zeta]$ gives a trivial conserved local
$p$-current. The second term gives $\rho[\psi; h[\psi;
e_\varphi[\psi]]]$, which is linear in $e_\varphi[\psi]$ and hence
trivial. The remaining term $\rho[\psi; e_\varphi[g[\psi]]]$ is trivial
for a slightly non-obvious reason: a deformation of the form
$e_\varphi[g[\psi]]$ necessarily comes from a local field redefinition
of the form $\psi \to \psi + tg[\psi] + O(t^2)$. Namely, recall that we
can always write $\rho[\psi;e_\varphi[\psi]] = \d k[\psi]$ for some
conserved local $(p+1)$-current $k$ and consider the identity
\begin{align}
	\d k[\psi + tg[\psi]]
		&= \rho[\psi + tg[\psi]; e_\varphi[\psi+tg[\psi]]] \\
		&= t \rho[\psi + tg[\psi]; e_\varphi[g[\psi]]]
			+ \rho[\psi + tg[\psi]; e_\varphi[\psi]] \\
		&= t (\rho[\psi; e_\varphi[g[\psi]]]
			+ \rho^{(1)}[\psi; e_\varphi[\psi]]) + O(t^2) .
\end{align}
When expanded in powers of $t$, all coefficients on the left-hand side
are exact. This shows that, up to the addition of the trivial conserved
local $p$-current $\rho^{(1)}[\psi;e_\varphi[\psi]]$, the term
$\rho[\psi; e_\varphi[g[\psi]]]$ is exact and hence trivial. This
concludes the proof.
\end{proof}

\begin{rem}
It is interesting to note that the deformation current mapping
\begin{equation}
	([\rho],[f])\mapsto [j_{\rho,f}],
\end{equation}
defined in Theorem~\ref{thm:def-cur}, identifies a selection criterion
on consistent deformations. We could say that a consistent deformation
$f$ is \emph{conservative} or \emph{$\rho$-conservative} when the
deformation current $j_{\rho,f}$ is trivial (on-shell exact). The
linearization obstruction generated by the deformation current
$j_{\rho,f}$ for a $\rho$-conservative deformation $f$, as to be
discussed in the next section, is necessarily trivial.
\end{rem}

\subsection{Obstructions from deformation currents}\label{sec:obstr-def-cur}
Consider the full non-linear PDE system $\E$ defined by the equation
form $e[\phi] = 0$, which we take to be locally regular, linearizable
and normal. Also, consider a linearizable background solution
$\varphi\in \S_\lin(\E)$. Recall that, a smooth $1$-parameter family of
solutions, $e[\phi_t] = 0$, of the form $\phi_t = \varphi + \psi_t$
satisfies Equation~\eqref{eq:lin-e}, that is,
\begin{equation}
	e_\varphi[\psi_t] = f_\varphi[\psi_t] .
\end{equation}
If $\psi_t = t\psi + O(t^2)$, then $\psi$ is a linearized solution,
$e_\varphi[\psi] = 0$, that can be extended to a smooth $1$-parameter
family of exact solutions, for instance $\phi_t$. Recall also that the
non-linearity $f_\varphi$ defines a representative of a consistent
deformation class $[f^{(m)}]\in H^0_\Def(e_\varphi)$ given by the Taylor
expansion
\begin{equation}
	f_\varphi[t\psi] = t^m f^{(m)}[\psi] + O(t^{m+1})
\end{equation}
for some $m>1$, where $f^{(m)}[\psi]$ is homogeneous of degree $m$ in
$\psi$. The main result of this paper, to be proven below, is that each
deformation current defined by $f^{(m)}$ canonically defines a
linearization obstruction.

As an intermediate step before formulating the main result, we need to
introduce a grading by degree of polynomial dependence on field bundle
sections on the spaces of local sections and local differential
operators. Let $^{(l)}\Secs_F(H) \sso \Secs_F(H)$ denote the subspace of
local sections of a vector bundle $H\to M$ of \emph{homogeneous
polynomial degree $l$}, that is, such that $h\in {}^{(l)}\Secs_F(H)$
only if $h[\psi]$ depends polynomially on $\psi$ and its derivatives,
and is homogeneous in $\psi$ of degree $l$.  Note that the horizontal
differential $\dh$, as well as all the differential operators in the
linearized Noether complex~\ref{eq:lNcpx}, the adjoint linearized
Noether complex~\ref{eq:alNcpx} and the complex $\hSecs_F(H)\to
\Secs_F(H) \to \Secs_\E(H)$ all preserve subspaces of homogeneous
polynomial degrees. This means that any of the cohomology spaces that we
have defined also have subspaces of homogeneous polynomial degrees; the
corresponding subspaces are denoted by $^{(l)}H^*_-(-) \sso H^*_-(-)$.
The representatives of each of these
subspaces can always be chosen of the same homogeneous polynomial
degree.

We are finally ready to formulate and prove the main theorem of this
paper, which relates the de~Rham cohomology $H^*_\dR(M)$ and the space
$H^*_\cosym(\E_\varphi)$ of rigid on-shell cosymmetries with
linearization obstructions at $\varphi$, and thus with potential
linearization instability.
\begin{thm}\label{thm:obstr}
If $\psi\in \S(\E_\varphi)$ is a linearized solution that can be
extended to a smooth $1$-parameter family of exact solutions, then it
must necessarily satisfy the conditions $Q^l_p(\psi) = 0$, $0\le p < n$,
$0\le l$, where
\begin{equation}
	Q^l_p\colon \S(\E_\varphi) \to
		{}^{(l)}H^p_\cosym(e_\varphi)^* \otimes H^{n-p}_\dR(M)
\end{equation}
are linearization obstructions of order $l+m$ defined by the
non-linearity $f_\varphi$, where $^*$ denotes the linear dual.
\end{thm}
\begin{proof}
First, recall the isomorphisms ($0\le p< n$)
\begin{equation}
	H^p_\cosym(e_\varphi) \cong H^{p+1}_\cur(\E_\varphi) \cong 
	H^{p}_\src(\E_\varphi)
\end{equation}
and note that they preserve subspaces of homogeneous polynomial degree.
That is, we can consider a null local $p$-source class $[\rho]\in
{}^{(l)}H^p_\src(\E_\varphi)$, equally well, to be an element $[\rho] \in
{}^{(l)}H^p_\cosym(e_\varphi)$. Using the on-shell consistent deformation class
$[f^{(m)}_\varphi]\in H^0_\Def(e_\varphi)$ and the map $j$ defined in
Theorem~\ref{thm:def-cur}, we obtain a conserved local $p$-current class
\begin{equation}
	j([\rho],[f^{(m)}_\varphi]) = [j_{\rho,f^{(m)}_\varphi}]
		\in {}^{(l+m)}H^p_\cur(\E_\varphi)
		\cong {}^{(l+m)}H^{n-p}_\Char(\E_\varphi) .
\end{equation}
The representative $j_{\rho,f^{(m)}_\varphi} \in \Omega^{n-p,0}(F)$ is a
local horizontal $(n-p)$-form and hence naturally defines a differential
operator $j_{\rho,f^{(m)}_\varphi} \colon \Secs(F) \to \Omega^{n-p}(M)$.
By construction, we know that $j_{\rho,f^{(m)}_\varphi}[\psi]$ is of
homogeneous polynomial degree $l+m$ and
that $\d j_{\rho,f^{(m)}_\varphi}[\psi] = 0$ if $\psi$
is a linearized solution. Hence it represents the de~Rham cohomology
class in $[\, j_{\rho,f^{(m)}_\varphi}[\psi] \, ] \in H^{n-p}_\dR(M)$.
Finally, we define the map $Q^l_p$ by the formula
\begin{equation}
	[\rho] \cdot Q^l_p(\psi) = [\, j_{\rho,f^{(m)}_\varphi}[\psi] \, ] ,
\end{equation}
where the dot on the left-hand side stands for the natural pairing
between $^{(l)}H^p_\cosym(e_\varphi)$ and its dual
$^{(l)}H^p_\cosym(e_\varphi)^*$.

Since $Q^l_p(\psi)$ is built out of a bilinear pairing between the
differential operator $\rho$ and $f^{(m)}_\varphi[\psi]$, the order
$Q^l_p(\psi)$ must be exactly $l+m$.

To conclude the proof, it is sufficient to show that $[\rho]\cdot
Q^l_p(\psi) = [0]$ for $\psi$ that is extendable to a smooth $1$-parameter
family of exact solutions $\phi_t = \varphi + \psi_t$, with $\psi_t =
t\psi + O(t^2)$. This family of exact solutions will satisfy the
equation $e_\varphi[\psi_t] = f_\varphi[\psi_t]$. Recall also that we
can find a conserved local $(p+1)$-current $k\in \Omega^{n-p-1,0}(F)$
such that $\rho[\Psi,e_\varphi[\Psi]] = \d k[\Psi]$ for arbitrary
$\Psi\in \Secs(F)$ and, in particular, this formula is still valid if we
set $\Psi = \psi_t$. We then have the following equalities in terms of
de~Rham cohomology classes:
\begin{align}
	[\, \rho[\psi_t;f_\varphi[\psi_t]] \,]
		&= t^{(l+m)} [\, \rho[\psi; f^{(m)}_\varphi[\psi]] \,] + O(t^{l+m+1})\\
		&= [\, \rho[\psi_t;e_\varphi[\psi_t]] \,]
		= [\, \d k[\psi_t] \,] \\
		&= [0] .
\end{align}
Note that comparing the coefficient of $t^{l+m}$ in these equations
gives precisely the desired equality $[\rho]\cdot Q^l_p(\psi) = 0$.
\end{proof}

\begin{rem}
The above theorem only shows how to canonically construct potential
linearization obstructions. It does not mean that the obstruction
$Q^l_p$ is necessarily non-trivial. For instance, if one can show that
the conserved current $j_{\rho,f^{(m)}_\varphi}$ is trivial, $[j_{\rho,
f^{(m)}_\varphi}] = [0] \in H^p_\cur(\E_\varphi)$ (equivalently, that
the consistent deformation $f^{(m)}$ is $\rho$-conservative, as defined
in the preceding section), then $[\rho]\cdot Q^l_p(\psi) = [0]$ is
trivial and does not pose any obstruction to extending linearized
solutions to exact ones. Further work must be done on a case by case
basis to show that the obstruction $Q^l_p$ is non-trivial.
Alternatively, one can use the triviality of $[\rho]\cdot Q^l_p(\psi)$,
checked in the way just given above, as a condition to select non-linear
consistent deformations like $f^{(m)}_\varphi$ that will not create
linearization obstructions (at least not of the kind constructed in
Theorem~\ref{thm:obstr}).
\end{rem}

Theorem~\ref{thm:obstr} is a significant generalization of the results
obtained in the literature that followed up the initial work of
\cite{brill-deser,fima-linstab}. The known linearization obstructions have
only been obtained in the case of either compact manifolds for equations
in Riemannian geometry or in the case of manifolds with compact Cauchy
surfaces for relativistic gauge theories. In light of our result, these
situations are immediately recognizable as obstructions coming
respectively from $H^n(M)\ne 0$ and from $H^{n-1}(M)\ne 0$, and also why
only gauge theories are susceptible in the latter case. Since the
equations considered in the literature have been stage-$0$ or stage-$1$
irreducible, our result also shows why similar obstructions are absent
for manifolds with $H^{n-p}(M)\ne 0$ and $p>1$. To see obstructions due
to these lower cohomology groups, one needs to consider non-linear
equations with higher stage reducibility and they are relatively
infrequent. Such gauge theories do appear in some specialized physics
literature, with non-linear $p$-form electromagnetism and some
supergravities as examples. They have also recently attracted
significant attention in the mathematics literature as \emph{higher
gauge theories}~\cite{baez-schr,schr-aksz}.

Finally, the calculations involved in identifying these linearization
obstructions have been rather cumbersome. This is especially the case
for relativistic gauge theories, where the necessary calculations lose
much of their geometric character due to non-covariant restriction to
some initial data surface. Our result, on the other hand, was obtained
completely geometrically and explains the covariance of the final
result. Famously, for Einstein equations, linearization obstructions
are only present if the background solution possesses non-trivial
Killing vectors (or satisfies the Killing condition). In particular, in
the existing work on relativistic gauge theories, significant effort was
needed in each case to obtain the analog of the Killing condition. Our
result, on the other hand, shows that this condition is precisely
Equation~\ref{eq:co-kil} defining a non-trivial cosymmetry, which for
obvious reasons we have also called the (generalized) co-Killing
condition.

The relation of our result to previous work is discussed in more detail
next.

\subsection{Integrated charges, relation with previous work}\label{sec:charge}
Note that the linearization obstruction $[\rho]\cdot Q^l_p(\psi)$
constructed in the preceding section is not strictly of the kind that
has been constructed in the existing literature on Einstein, Yang-Mills
and related equations. While both formulations pass through conserved
currents, we have given the obstructions as valued in the de~Rham
cohomology $H^*_\dR(M)$ classes of these currents, while the usual
formulation gives them in terms of corresponding integrated charges. An
integrated charge is obtained by integrating a conserved current over a
closed compact submanifold.

Our formulation can give rise to integrated charges as well. If $\Sigma
\sso M$ is a closed compact submanifold, then
\begin{equation}
	\langle [\Sigma] , [\rho] \cdot Q^l_p(\psi) \rangle
		= \int_\Sigma j_{\rho,f^{(m)}_\varphi}[\psi]
\end{equation}
is the corresponding integrated charge. The charges obtained in this way
are not all independent. Each charge depends only on the (singular)
homology class $[\Sigma]\in H_*(M)$. Moreover, linear combinations of
homology classes lead to linear combinations of charges. Therefore, to
obtain a set of independent charges, we need to pick a basis $[\Sigma_i]$ for
$H_*(M)$ and pair each basis element with $[\rho]\cdot Q^l_p(\psi)$. But
this amounts to no more than composing $Q^l_p$ with the isomorphism
$H^*_\dR(M) \cong \bigoplus_i \mathbb{R}^{b_i}$ defined by the basis
$[\Sigma_i]$, where the $b_i = \dim H_i(M)$ are the Betti numbers of $M$. The
fact that this map is an isomorphism is simply a restatement of
Poincar\'e duality~\cite{wiki-pd}. So, an integrated charge over a
non-trivial closed compact submanifold $\Sigma\sso M$ is simply a
witness to the existence of a non-trivial cohomology class in $H^*(M)$.
In other words, our formulation of linearization obstructions is
equivalent to the usual one.

\subsection{Asymptotic boundary conditions}\label{sec:asymp}
On the other hand, our formulation is more convenient in the discussion
of asymptotic boundary conditions on non-compact manifolds and
corresponding non-standard de~Rham cohomology. For instance, if $M$ is
non-compact, we may consider a suitable subspace of $\tSecs(F)\sso
\Secs(F)$ of the space sections of the field bundle $F\to M$, along with
a corresponding refinement $\tOmega^k(M)\sso \Omega^k(M)$ of the de~Rham
complex, with $\d \tOmega^k(M) \sse \tOmega^{k+1}(M)$. These subspaces
could be selected by imposing some asymptotic boundary conditions at the
open ends of $M$. An extreme example would be require all sections and
forms to have compact support. The cohomology $\tilde{H}^*_\dR(M) =
H(\tOmega^*(M),\d)$ could be different from the standard $H^*_\dR(M)$.
Consider, with respect to the linearized PDE system, a conserved local
current class $[k]$ and the corresponding null local source class $[\rho]
= [\dh k]$. If the local horizontal forms $\rho$ and $k$ can be chosen
such that the boundary conditions on $\psi\in \tSecs(F)$ imply that
$\rho[\psi], k[\psi]\in \tOmega^*(M)$. The construction of the
linearization obstruction $[\rho]\cdot Q^l_p(\psi)$ still works, but the
result is now valued in the non-standard de~Rham cohomology
$\tilde{H}^*_\dR(M)$ rather than $H^*_\dR(M)$. As before, it is still a
non-trivial problem to check that the resulting potential linearization
instability $Q^l_p$ is non-trivial. However, if the boundary conditions
are chosen such that $\tilde{H}^*(M) = 0$, then the linearization
obstruction yielded by this construction is necessarily trivial.

The above logic appears to be the reason behind the absence of
linearization obstructions for the common choice of asymptotically flat
boundary conditions for the Einstein equations~\cite{cb-deser,cbfm}.

\section{Examples}\label{sec:ex}
The PDE systems studied in the physics literature are mostly variational
(coming from classical Lagrangian field theories). These are also the
kind of systems analyzed in previous work on linearization instabilities.
Our analysis is applicable to a more general class of PDE systems,
including non-variational one. It would be nice to have explicit
examples of theories from each of the corners missed by the existing
literature. On the other hand, it is rather easy to provide examples of
non-variational PDE systems by taking a variational system, $e[\psi] =
0$, and pre-composing with an arbitrary differential operator, $e[g[\eta]]
= 0$. The resulting system will generically no longer be variational,
but will possess essentially the same Noether complex as the original
one. Thus, the essentially new examples that we give below are
restricted to $1$-dimensional systems (ODEs) and higher gauge theories.

\subsection{Stage-$0$ irreducible systems}

\subsubsection{ODEs}
Consider an ordinary differential equation (ODE) that is defined on a circle,
$M=S^1$, and scalar valued, $G=F=M\times \mathbb{R}$:
\begin{equation}
	\ddot{\phi} = f[\phi] ,
\end{equation}
with $\dot{\phi} = \d\phi/\d{t}$ and $t$ a coordinate on $M$. This is
essentially the $1$-dimensional particle equation with cyclic time and
force term $f$, which, say, is homogeneous of polynomial degree $>1$. If
we linearize about $\varphi=0$, the linearized equation is just the free
particle equation $\ddot{\phi}=0$ with cyclic time. Its only solutions
are constants, $\S(\E_\varphi) = \{\psi(t)\mid \psi(t) = \psi_0 \} \cong
\mathbb{R}^1$. It is straightforward to check that $\tsxi{}=\d{t}$ is a
non-trivial cosymmetry with corresponding conserved $1$-current $k[\psi]
= \dot{\psi}$ (the momentum) and null source $\rho[\psi] =
\ddot{\psi}\,\d{t}$.

If the non-linearity is $f_\varphi[\psi] = f[\psi] = \psi^2$,
the corresponding deformation $0$-current is
\begin{equation}
	j_{\rho,\psi^2}[\psi] = \psi^2\, \d{t} ,
\end{equation}
which is easily seen to be non-trivial. In fact, the integrated charge
\begin{equation}
	\langle [S^1],Q_\rho(\psi) \rangle
	= \int_{S^1} j_{\rho,\psi^2}[\psi]
	= c \psi_0^2 \ne 0,
\end{equation}
with $c$ the circumference of $S^1$, for any linearized solution
$\psi(t) = \psi_0 \ne 0$. Clearly, the zero set $Q_\rho(\psi) = 0$
consists of a single point, $\{\psi(t)\mid \psi(t)=0\} \cong
\mathbb{R}^0$.

On the other hand, if we consider $\tsxi{} = \dot{\psi}\, \d{t}$,
$j[\psi] = \frac{1}{2}\dot{\psi}^2$ (the energy) and $\rho[\psi] =
\dot{\psi}\ddot{\psi}\,\d{t}$, the corresponding deformation current is
trivial:
\begin{equation}
	j_{\rho,\psi^2}[\psi] = \dot{\psi} \psi^2\,\d{t}
	= \frac{1}{3} \d\psi^3 .
\end{equation}
In other words, the force term $f[\phi] = \phi^2$ is energy-conservative.

\subsubsection{Semilinear elliptic PDEs}
A very similar situation occurs on higher dimensional compact
manifolds, say $M=S^2$, with Laplace-type semilinear equations, $F =
M\times \mathbb{R}$, $G= \tilde{F}^* = \Lambda^2M$,
\begin{equation}
	{*}(\Delta \phi + l(l+1)\phi) = f[\phi] ,
\end{equation}
where we interpret $S^2$ as the standard, round unit sphere, $*$ is the
Hodge star, $\Delta = {*}\d{*}\d + \d{*}\d{*}$ is the Laplacian on it
and $l\ge 0$ is an integer. The linearized equation about $\varphi=0$ is
${*}(\Delta \psi + l(l+1)\psi)=0$ and its solution space is well known
to be of dimension $2l+1$ and spanned by spherical harmonics,
$\S(\E_\varphi)\cong \mathbb{R}^{2l+1}$. A non-trivial cosymmetry is
given by a spherical harmonic $\tsxi{} = Y^l_m$. The corresponding
conserved $1$-current is $k_m[\phi] = {*}(Y^l_m\d\phi-\phi\d Y^l_m)$ and
corresponding null $0$-source is $\rho_m[\phi] =
{*}Y^l_m(\Delta\phi+l(l+1)\phi)$.

If we set $f_\varphi[\psi] = f[\psi] = \psi^2$, we get the non-trivial
deformation current
\begin{equation}
	j_{\rho_m,\psi^2} = {*}Y^l_m\psi^2 ,
\end{equation}
with integrated charge
\begin{equation}
	\langle [S^2], Q_m(\psi) \rangle
	= \int_{S^2} {*} Y^l_m \psi^2 .
\end{equation}
In other words, in order to satisfy all linearization obstructions of
the form $Q_m(\psi) = 0$, $-l\le m \le l$, the decomposition of $\psi^2$
into spherical harmonics must have vanishing coefficients. These
conditions are only satisfied by the zero solution; these
obstructions drop the dimension of the linearized solution space from
$2l+1$ to $0$, $\{\psi \mid \psi = 0\} \cong \mathbb{R}^0$.

\subsubsection{Other examples}
Other examples considered in the existing literature concern some
problems from Riemannian geometry on compact manifolds, such as the
equation describing metrics of constant scalar curvature~\cite{fima-riem}.
The identified obstructions can also be computed using our main result.

\subsection{Stage-$1$ irreducible systems}

\subsubsection{Einstein equations}
Einstein's equations describe the dynamics of the gravitational field, a
Lorentzian metric, on a manifold $M$ either in vacuum or, with
appropriate modification, in the presence of matter. A Lorentzian metric
$g\in \Secs(S^2T^*M)$ is a section of the bundle $F=S^2T^*M\to M$ of
rank-$2$ covariant symmetric tensors. Detailed formulas necessary to
exhibit the structure of the equations, the non-linearity and the
linearization obstructions are rather lengthy and, for our purposes, not
particularly illuminating. The linearization instabilities of Einstein
equations have been studied extensively, they in fact gave birth to this
subject, and all the relevant details are summarized in the introductory
sections of~\cite{struc1}. Below, we merely indicate how they fit into
the framework of our main result.

The linearized Noether complex for Einstein equations, as well as its
formal adjoint, about a background metric $\varphi=\bar{g}$, are given by
\begin{align}
	&\xymatrix{
		\Secs_F(S^2T^*M) \ar[rr]^-{e_\varphi=L} &&
		\Secs_F(S^2T^*M) \ar[r]^-{B} &
		\Secs_F(T^*M) \ar[r] &
		0 ,
	} \\
	&\xymatrix{
		\Secs_F(S^2T^*M) \ar@{<-}[rr]^-{e_\varphi^*=L} &&
		\Secs_F(S^2T^*M) \ar@{<-}[r]^-{K} &
		\Secs_F(T^*M) \ar@{<-}[r] &
		0 .
	}
\end{align}
Note that we have identified $\tilde{F}^* \cong F$, and the same for the
other tensor bundles, using $\bar{g}$ to raise and lower tensor indices
as well as to construct a canonical volume density. Here, $L$ is the
differential operator of the linearized Einstein equations, also known
as the \emph{Lichnerowicz operator}. The operators $B$ and $K$
correspond, respectively, to the linearized \emph{Bianchi identity} and
\emph{Killing equation}. In local coordinates, they are
\begin{align}
	(B[h])_{i} &= \nabla^j h_{ij} , \\
	(K[v])_{ij} &= \frac{1}{2}(\nabla_i v_j + \nabla_j v_i) ,
\end{align}
where $\nabla$ is the Levi-Civita connection with respect to $\bar{g}$,
which is also used to raise and lower tensor indices. Solutions to the
Killing equation are called \emph{Killing vectors}; in our framework
they are the rigid stage-$1$ cosymmetries of the linearized Einstein
equations. The conserved $2$-currents corresponding to these Killing
vectors are known as Abbott-Deser fluxes~\cite{abbott-deser} and the
corresponding null $1$-sources do not have a name in the literature. On
the other hand, the corresponding deformation $1$-currents, constructed
using the quadratic term in Einstein's equations, are known as the Taub
conserved currents~\cite{taub}. The connection between the presence of
non-trivial Killing vectors and the vanishing of their Taub charges as a
linearization obstruction was first noticed by
Moncrief~\cite{moncrief-killing2}.

\subsubsection{Yang-Mills equations}
The notational background for this section is given in
Appendix~\ref{sec:forms}. In Yang-Mills theory~\cite{wiki-ym}, the basic
dynamical field is a semi-simple Lie algebra $\g$-valued $1$-form
$\alpha\in \Omega^1(M,\g)$, so $F=\g\otimes \Lambda^1M$. Where $M$ is a
manifold of $\dim M = n$ and endowed with a (pseudo-)Riemannian metric.
It is a stage-$1$ irreducible, non-linear deformation of the linear
Maxwell theory. Its Lagrangian density with the leading order
non-linearity, as a local variational form, is
\begin{equation}
	\mathcal{L} = -\frac{1}{4}\langle \dh\alpha \wedge {*}\dh\alpha\rangle
		+ \frac{1}{2} \langle \dh\alpha \wedge {*}[\alpha\wedge \alpha] \rangle
		+ O(\alpha^4) .
\end{equation}
The linearized equations, about $\varphi = \alpha = 0$, and the leading
order non-linear consistent deformation are obtained by a vertical
variation of the Lagrangian density:
\begin{align}
	\dv\mathcal{L}
	&= -\frac{1}{2}\langle \dv\dh\alpha \wedge {*}\dh\alpha \rangle
		+ \frac{1}{2}\langle \dv\dh\alpha \wedge {*}[\alpha\wedge\alpha] \rangle
		+ \langle \dh\alpha \wedge {*}[\dv\alpha\wedge\alpha] \rangle \\
\notag & \qquad {}
		+ O(\alpha^4) \\
	&= -\frac{1}{2}\langle \dv\alpha \wedge {*}\Dh\dh\alpha \rangle
		+ \frac{1}{2}\langle \dv\alpha\wedge {*}\Dh[\alpha\wedge\alpha] \rangle
		+ \langle \dv\alpha \wedge [\alpha \wedge {*}\dh\alpha] \rangle \\
\notag & \qquad {}
		+ \dh({\cdots}) + O(\alpha^4) .
\end{align}
We can read off the direct and adjoint linearized Noether complexes as
\begin{align}
\label{eq:ym-lNcpx}
	&\xymatrix{
		\Omega^{1,0}(F,\g) \ar[rr]^-{e_\varphi={*}\Dh\dh} &&
		\Omega^{3,0}(F,\g) \ar[r]^-{\dh} &
		\Omega^{4,0}(F,\g) \ar[r] &
		0 ,
	} \\
\label{eq:ym-alNcpx}
	&\xymatrix{
		\Omega^{3,0}(F,\g) \ar@{<-}[rr]^-{e_\varphi^*={*}\Dh\dh} &&
		\Omega^{1,0}(F,\g) \ar@{<-}[r]^-{\dh} &
		\Omega^{0,0}(F,\g) \ar@{<-}[r] &
		0 ,
	}
\end{align}
while the leading order consistent deformation is
\begin{equation}
	f^{(2)}_\varphi[\alpha] = f[\alpha]
	= {*}\delta\frac{1}{2}[\alpha\wedge\alpha]
		+ [\alpha\wedge {*}\d\alpha] .
\end{equation}
A stage-$1$ rigid cosymmetry is any $\g$-valued $0$-form $\eps \in
\Omega^{0,0}(F,\g)$ such that $\d\eps = 0$. The corresponding conserved
$2$-current and null $1$-source are, respectively, $k = \langle \eps
\wedge {*}\d \alpha\rangle$ and $\rho = \langle \eps \wedge
{*}\delta\d\alpha\rangle$ (up to signs). The resulting deformation
current, which we can check is not off-shell closed, is
\begin{align}
	j_{\rho,f}[\alpha]
		&= \langle \eps \wedge [\alpha\wedge {*}\d\alpha] \rangle, \\
	\d j_{\rho,f}[\alpha]
		&= -\langle \eps \wedge [\alpha\wedge (\d{*}\d\alpha)] \rangle ,
\end{align}
where the neglected term is trivial because of the identity $\d{*}\delta
= 0$.

The above construction gives an obstruction $Q_{\eps}$ for each
rigid cosymmetry $\eps\in \Omega^{0,0}(F,\g)$ valued in
$H^{n-1}_\dR(M)$. These are precisely the obstructions that were
previously obtained by Moncrief~\cite{moncrief-ym}, where he checked
that they are non-trivial and also sufficient. Moncrief, like other
existing literature, implicitly used the existence of a compact Cauchy
surface as a witness to the non-triviality of $H^{n-1}_\dR(M)$.
Obstructions appear also at non-vanishing backgrounds $\varphi=A$, where
the rigid cosymmetries are geometrically identified with $A$-parallel
$\g$-valued scalars. If $A$ is interpreted as a connection on a
principal bundle, these parallel scalars are intimately connected with
its the holonomy group.

\subsubsection{Chern-Simons}
The notational background for this section is given in
Appendix~\ref{sec:forms}. In Chern-Simons theory~\cite{wiki-cs}, where
$\dim M = 3$, the basic dynamical field is also a semi-simple Lie
algebra $\g$-valued $1$-form $\alpha\in \Omega^1(M,\g)$, so $F=\g\otimes
\Lambda^1M$. It is a stage-$1$ irreducible, non-linear theory with
Lagrangian density, as a local variational form,
\begin{equation}
	\mathcal{L} = \frac{1}{2}\langle \alpha\wedge \dh\alpha \rangle
		- \frac{1}{3} \langle \alpha\wedge [\alpha \wedge \alpha] \rangle .
\end{equation}
The linearized equations, about $\varphi = \alpha = 0$, and the
non-linear consistent deformation are obtained by a vertical variation
of the Lagrangian density:
\begin{align}
	\dv\mathcal{L}
	&= \langle \dv\alpha \wedge \dh\alpha \rangle
		- \langle \dv\alpha \wedge [\alpha \wedge \alpha] \rangle
		+ \dh({\cdots}) .
\end{align}
We can read off the direct and adjoint linearized Noether complexes as
\begin{align}
\label{eq:cs-lNcpx}
	&\xymatrix{
		\Omega^{1,0}(F,\g) \ar[rr]^-{e_\varphi=\dh} &&
		\Omega^{2,0}(F,\g) \ar[r]^-{\dh} &
		\Omega^{3,0}(F,\g) \ar[r] &
		0 ,
	} \\
\label{eq:cs-alNcpx}
	&\xymatrix{
		\Omega^{2,0}(F,\g) \ar@{<-}[rr]^-{e_\varphi^*=\dh} &&
		\Omega^{1,0}(F,\g) \ar@{<-}[r]^-{-\dh} &
		\Omega^{0,0}(F,\g) \ar@{<-}[r] &
		0 ,
	}
\end{align}
while the leading order consistent deformation is
\begin{equation}
	f^{(2)}_\varphi[\alpha] = f[\alpha]
	= [\alpha \wedge \alpha] .
\end{equation}
Just as in the case of Yang-Mills theory, a rigid stage-$1$ cosymmetry
is any $\g$-valued $0$-form $\eps \in \Omega^{0,0}(F,\g)$ such that
$\d\eps = 0$. The corresponding conserved $2$-current and null
$1$-source are, respectively, $k = \langle \eps \wedge \alpha \rangle$
and $\rho = \langle \eps \wedge \d\alpha\rangle$. The resulting
deformation current, which we can check is not off-shell closed, is
\begin{align}
	j_{\rho,f}[\alpha]
		&= \langle \eps \wedge [\alpha\wedge\alpha] \rangle, \\
	\d j_{\rho,f}[\alpha]
		&= -\langle \eps \wedge [\alpha\wedge (\d\alpha)] \rangle .
\end{align}
The geometric interpretation of the rigid cosymmetries, and their
connection with the holonomy group, remains the same as in Yang-Mills
theory.

The solutions of the Chern-Simons equations constitute flat connections
on the corresponding principal bundle over $M$ (in the above simplified
setting, the principal bundle is trivial). The moduli space of flat
connections, that is, the solution space we have been considering modulo
the gauge transformations $\alpha \to \alpha + \d\omega$, has been
studied extensively. It is known that this moduli space is not a smooth
manifold, because it may have quadratic algebraic singularities at
connections with non-trivial holonomy groups~\cite{goldman}. That result
is entirely consistent with the obstructions computed above.

\subsection{Other examples}
Other irreducible gauge theories that have been considered in the
existing literature include Einstein-Yang-Mills~\cite{arms-eym} and classical
$N=1$ supergravity~\cite{bao-sugra}. The obstructions and (co-)Killing
conditions identified in these theories are reproduced by our framework
and have been checked to be non-trivial and sufficient. See
\cite{struc1,struc2} for a historical discussion and detailed
references.

\subsection{Stage $>1$ irreducible systems}

\subsubsection{Non-abelian Freedman-Townsend $2$-form}
The notational background for this section is given in
Appendix~\ref{sec:forms}. Generalizations of Maxwell equations to
$p$-forms (\emph{abelian} $p$-forms) provide a standard source of
examples of higher stage irreducible linear PDE systems. Their
consistent deformations (\emph{non-abelian} $p$-forms) provide examples
of corresponding higher stage irreducible non-linear PDE systems.

On a manifold $M$ of $\dim M = 4$ that is endowed with a
(pseudo-)Riemannian metric, the only consistent, non-linear deformation
of abelian $2$-form field theory is the Freedman-Townsend
model~\cite{freedman-townsend}. The basic dynamical field is a
semi-simple Lie algebra $\g$-valued $2$-form $\beta \in \Omega^2(M,g)$,
so $F = \g\otimes \Lambda^2 M$. Its Lagrangian density, with the leading
order non-linearity, as a local variational form, is
\begin{equation}
	\mathcal{L}[\beta] = -\frac{1}{8}\langle \dh\beta\wedge{*}\dh\beta \rangle
		- \frac{1}{4} \langle \beta\wedge [{*}\dh\beta\wedge {*}\dh\beta] \rangle
		+ O(\beta^4) .
\end{equation}
The linearized equations, about $\varphi = \beta = 0$, and the
non-linear consistent deformation are obtained by a vertical variation
of the Lagrangian density:
\begin{align}
	\dv\mathcal{L}[\beta]
	&= -\frac{1}{4} \langle \dv\dh\beta\wedge {*}\dh\beta \rangle
		- \frac{1}{4} \langle \dv\beta \wedge[{*}\dh\beta\wedge {*}\dh\beta]
				\rangle \\
\notag & \qquad {}
		+ \frac{1}{2} \langle \beta\wedge [{*}\dh\beta \wedge {*}\dv\dh\beta]
				\rangle + O(\beta^4) \\
	&= -\frac{1}{4} \langle \dv\beta \wedge {*}\Dh\dh\beta \rangle
		- \frac{1}{4} \langle \dv\beta \wedge[{*}\dh\beta\wedge {*}\dh\beta]
				\rangle \\
\notag & \qquad {}
		+ \frac{1}{2} \langle \dv\beta \wedge {*}\Dh[\beta\wedge{*}\dh\beta]\rangle
		+ O(\beta^4) .
\end{align}
We can read off the direct and adjoint linearized Noether complexes as
\begin{align}
	&\xymatrix{
		\Omega^{2,0} \ar[rr]^-{e_\varphi={*}\Dh\dh} &&
		\Omega^{2,0} \ar[r]^-{\dh} &
		\Omega^{3,0} \ar[r]^-{\dh} &
		\Omega^{4,0} \ar[r] &
		0 } , \\
	&\xymatrix{
		\Omega^{2,0} \ar@{<-}[rr]^{e_\varphi^*={*}\Dh\dh} &&
		\Omega^{2,0} \ar@{<-}[r]^{-\dh} &
		\Omega^{1,0} \ar@{<-}[r]^{\dh} &
		\Omega^{0,0} \ar@{<-}[r] &
		0 } ,
\end{align}
while the leading order consistent deformation is
\begin{equation}
	f^{(2)}_\varphi[\beta] = f[\beta]
	= [{*}\d\beta \wedge {*}\d\beta]
		- {*}\delta (2[{*}\d\beta \wedge \beta]) .
\end{equation}
A rigid stage-$2$ cosymmetry is any $\g$-valued $0$-form $\eps\in
\Omega^{0,0}(F,\g)$ such that $\d\eps = 0$. The corresponding conserved
$3$-current and null $2$-source are, respectively $k=\langle \eps\wedge
{*}\d\beta \rangle$ (a $1$-form) and $\rho = \langle \eps \wedge
{*}\delta\d\beta \rangle$ ($2$-form) (up to signs). The resulting
deformation current, which we can check is not off-shell closed, is
\begin{align}
	j_{\rho,f}[\beta]
		&= \langle \eps \wedge [{*}\d\beta \wedge {*}\d\beta] \rangle , \\
	\d j_{\rho,f}[\beta]
		&= \langle \eps \wedge [\d{*}\d\beta\wedge {*}\d\beta] \rangle ,
\end{align}
where the neglected term is trivial because of the identity $\d{*}\delta
= 0$.

Therefore, for each such rigid stage-$1$ cosymmetry $\eps$, we have
found a second order linearization obstruction for the background
solution $\varphi = \beta = 0$ of the Freedman-Townsend model on any
(pseudo-)Riemannian $4$-manifold with non-vanishing
$H^2_\dR(M)$:
\begin{equation}
	[\eps]\cdot Q_\rho \colon \S(\E_\varphi) \to H^2_\dR(M) .
\end{equation}
An example of a Lorentzian, globally hyperbolic spacetime is the
Schwarzschild black hole spacetime, whose topology is
$\mathbb{R}^2\times S^2$. This observation appears to be new. We have
not explicitly checked that these obstructions are non-trivial, but that
is likely to be the case, by analogy with the Yang-Mills and
Chern-Simons examples. It would be interesting to obtain an
interpretation for these rigid cosymmetries in terms of the higher
bundle interpretation of higher gauge theories~\cite{baez-schr}.

\subsubsection{Other examples}
Higher gauge theories have attracted a lot of attention recently in some
mathematical literature (for instance, see~\cite{baez-schr,schr-aksz,royt}). Many of these theories involve Lie
algebra $\g$-valued $p$-forms as dynamical fields and exhibit
non-linearities similar to those of Yang-Mills, Chern-Simons and
Freedman-Townsend theories. In addition, $N>1$ classical supergravity
theories also involve dynamical $p$-forms, non-linearly coupled to other
fields~\cite{dewit}. It appears that the question of linearization
stability or instability of particular background solutions in these
theories has yet to attract any attention.

\section{Discussion}\label{sec:discuss}
We have considered the question of linearization instability for
a large class of non-linear PDE systems, which includes relativistic
Lagrangian field theories, with both irreducible and reducible gauge
theories among them. This class consists of all (not necessarily
Lagrangian) PDE systems, for which a Noether complex can be defined.

The question of linearization stability has two aspects. One is to
identify obstructions, which prevent an arbitrary linearized solution
from being extended to a family of exact ones. The other, once
sufficiently many obstructions have been identified, is to show that no
further obstructions exist. We have concentrated only on the first
aspect.  Moreover, we have restricted our attention only those
obstructions valued in the domain manifold's de~Rham cohomology, with
the cohomology representatives being constructed locally out of the
linearized solutions. Though, within this class of obstructions, we
likely give an exhaustive classification. (It is a matter of connecting
the result of Theorem~\ref{thm:def-cur} with the consistent deformations
coming from the non-linearity at leading, as was done in
Theorem~\ref{thm:obstr}, or higher orders). Interestingly enough, it
captures essentially all linearization obstructions that have been
identified in the work following up the original articles Brill~\&
Deser~\cite{brill-deser} and Fischer~\& Marsden~\cite{fima-linstab}. In
many cases these known obstructions have also been shown to be
sufficient~\cite{struc1,struc2}, at least when concentrating on initial
data and ignoring blow-up singularities.

Our main result (Theorem~\ref{thm:obstr}) provides a streamlined,
unified way of identifying such obstructions. Also, both the results and
intermediate calculations are everywhere geometric and covariant with
respect to the PDE domain. This is an improvement over previous
calculations, which treated individually each PDE system of interest and
often required some non-geometric intermediate steps. 

The conditions identifying linearization unstable backgrounds in known
examples have all been, in a sense, generalizations of the Killing
condition (existence of non-trivial Killing vectors) in general
relativity. Our main result puts this observation into a more general
context. We term the most general form of this condition the
\emph{co-Killing condition} (with the prefix \emph{co-} implying that it
is defined by the Noether complex, rather than the complex of gauge
generators, which are identical in Lagrangian theories). This condition
corresponds to the presence of non-trivial rigid (higher stage)
cosymmetries (Section~\ref{sec:cosym}). In Lagrangian theories (by a
generalization of Noether's second theorem), these are dual to rigid
(higher stage) symmetries, of which Killing vectors are in fact an
example.

Our result also clarifies the role of the non-trivial topology of the
PDE domain. The known linearization obstruction have so far appeared
mostly on PDE domains (say of dimension $n$) that are compact ($H^n_\dR
\ne 0$) or with a compact initial data surface ($H^{n-1}_\dR \ne 0$). No
particularly clear statement has been made about non-compact domains.
Now, from Theorem~\ref{thm:obstr}, we know that any non-trivial de~Rham
cohomology class in $H^*_\dR$ can generate a linearization obstruction.
However, the lower degree cohomology classes come into play only in the
presence of rigid higher stage cosymmetries. In particular, since
reducible gauge theories have not been analyzed in the literature on
linearization instabilities (Einstein, Yang-Mills and related theories
are \emph{irreducible} gauge theories), this explains why only the
$H^n_\dR$ and $H^{n-1}_\dR$ cohomology classes have played a role.
Furthermore, our result shows that linearization obstructions can appear
even in the absence of non-trivial topology, as long as the de~Rham
complex develops non-trivial cohomology when restricted by some
asymptotic boundary conditions on open manifolds
(Section~\ref{sec:asymp}).

Further, the intermediate result of Theorem~\ref{thm:def-cur} may be of
interest in its own right. It shows that consistent deformations of
linear PDE systems relate on-shell conserved currents of different
degrees. Namely, a consistent deformation and a conserved
$(p+1)$-current bilinearly define a $p$-current, which we have called a
\emph{deformation current}. We have named \emph{conservative}
those consistent deformations that have trivial
deformation currents. It may be interesting to analyze the subset of
conservative deformations and its relation to the general problem of
adding non-linear interactions to linear field
theories~\cite{henneaux-con}. In particular,
it is interesting to investigate whether the rank (in either argument)
of the deformation current mapping is expressible is some known
invariant of the linear PDE system.

Finally, in the problem of deformation quantization of classical
theories, the differential geometry of the phase space plays a crucial
role. When the phase space is a smooth manifold, the Fedesov
construction essentially solves the problem~\cite{fedosov}.
In classical field theories, the phase space can be identified
with the space of solutions. Non-smooth, singular points (precisely
those solutions that are linearization unstable) then pose an obstacle
to quantization. The situation with deformation quantization in the
presence of such singularities is much less clear, though some work has
been done in that direction~\cite{pflaum1,pflaum2,gotay}. We hope that
the results of this paper could serve as tools in this program.

\paragraph{Acknowledgements.}
The author thanks Glenn Barnich for enlightening discussions regarding
characteristic cohomology, Artur Sergyeyev for pointing out the concept
of cosymmetry, Sergey Igonin for references~\cite{verbov1,verbov2} and
Ian Anderson for reference~\cite{aa-taub}. The author also acknowledges
support from the Netherlands Organisation for Scientific Research (NWO)
(Project No.\ 680.47.413).

\appendix 

\section{Dynamical Forms}\label{sec:forms}
Below, we collect useful formulas for explicit calculation with theories
where the dynamical fields are differential $k$-forms on the manifold
$M$ of $\dim M = n$.  We take the field bundle to be $F=\Lambda^k M$,
often the equation bundle and the Noether bundles are bundles of forms over
$M$ as well, say $G\cong \Lambda^l M$ and $Z^i\cong \Lambda^{k_i} M$.
The space of $F$-local section can then be identified with the space of
local horizontal forms of the same degree, $\Secs_F(F) \cong
\Omega^{k,0}(F)$, $\Secs_F(G) \cong \Omega^{l,0}(F)$ and $\Secs_F(Z^i)
\cong \Omega^{k_i,0}(F)$. Dual densities can be identified with forms of
complementary degree, $\tilde{F}^* \cong \Lambda^{n-k}$, where the
natural fiber-wise pairing is realized as
\begin{equation}
	\lambda \cdot \mu = \lambda \wedge \mu ,
\end{equation}
with $\lambda\in \Secs(F)$ and $\mu\in \Secs(\tilde{F}^*)$. The de~Rham
differential $\d$ extended to act on local sections is none other than
the horizontal differential $\dh$ from the variational bicomplex. Using
the algebra of differential forms, it is easy to find its formal
adjoint:
\begin{equation}
	\d\lambda\wedge \mu - \lambda \wedge (-)^{|\lambda|+1}\d\mu
	= \d(\lambda\wedge \mu) ,
\end{equation}
which means that the formal adjoint of $\dh\colon \Omega^{l,0}(F) \to
\Omega^{l+1,0}(F)$ is $\dh^* = (-)^{l+1}\dh \colon \Omega^{n-l-1,0}(F)
\to \Omega^{n-k}$.

When the manifold $M$ is endowed with a (pseudo-)Riemannian metric $g$,
we can define the corresponding Hodge dual operator ${*}\colon
\Omega^{l,0}(F) \to \Omega^{n-l,0}(F)$. Recall the following
identities~\cite{wiki-hodge}:
\begin{align}
	\lambda_1 \wedge {*}\lambda_2 &= \lambda_2 \wedge {*}\lambda_1 , \\
	{*}{*} \lambda &= \epsilon_{|\lambda|} \eta , \quad
		\text{with}~ \epsilon_k = (-)^{k(n-k)+(n-s)/2} ,
\end{align}
where $s$ is the signature of $g$. If $n=4$ and the metric has
Lorentzian signature $({-}{+}{+}{+})$, $s=2$, then $\epsilon_k =
(-)^{k+1}$. With Hodge duality available, it is also possible to
identify the dual densities of $l$-forms with $l$-forms themselves with
the Hodge fiber-wise pairing
\begin{equation}
	\lambda\cdot \mu = \lambda \wedge {*}\mu .
\end{equation}
The formal adjoint of $\dh$ with respect to the above pairing is the
horizontal de~Rham \emph{codifferential} denoted by $\Dh$ (and by
$\delta$ when not acting on local horizontal forms):
\begin{equation}
	\dh\lambda \wedge {*}\mu - \mu\wedge {*}\Dh\mu = \dh(\mu\wedge {*}\nu) .
\end{equation}
It satisfies the identity $\Dh\lambda = (-)^{|\lambda|}
\epsilon_{|\lambda|} \dh$, so that $\Dh^2 = 0$ and $\dh{*}\Dh =
\Dh{*}\dh = 0$.

For calculations that include vertical forms, it is convenient to extend
the Hodge dual to all local variational forms so that
\begin{equation}
	{*}(\nu\wedge \lambda) = \nu \wedge {*}\lambda ,
\end{equation}
where $\nu$ is any purely vertical form and $\lambda$ is any purely
horizontal one. This way, the Hodge dual commutes with the vertical
differential, $\dv{*} = {*}\dv$. Moreover, the Hodge dual pairing
satisfies the identity
\begin{align}
	(\omega\wedge\lambda)\wedge{*}(\pi\wedge\mu)
	&= (-)^{|\lambda||\pi|} \omega\wedge\pi\wedge (\lambda\wedge{*}\mu) \\
	&= (-)^{(|\lambda|+|\omega|)|\pi|} \pi\wedge\omega\wedge
		(\mu\wedge{*}\lambda) \\
	&= (-)^{|\lambda||\pi|+|\omega||\pi|+|\mu||\pi|}
		(\pi\wedge\mu)\wedge{*}(\omega\wedge\lambda) ,
\end{align}
whenever $\omega$ and $\pi$ are purely vertical, while $\lambda$ and
$\mu$ are purely horizontal.

Let $\g$ be a Lie algebra. When the dynamical form fields admit the
interpretation of being components of a (higher) connection on a
(higher) principal bundle defined by $\g$ over $M$, it becomes
convenient to use $\g$-valued forms as dynamical fields, that is, $F =
\g\otimes \Lambda^k M$. The Lie algebra naturally comes with the
following bilinear operations:
\begin{align}
	\text{(commutator)} \quad [-]\colon & \g\otimes \g \to \g , \\
	\text{(Killing form)} \quad \langle-\rangle\colon &
		\g\otimes\g \to \mathbb{R} ,
\end{align}
where $[-]$ is antisymmetric, while $\langle-\rangle$ is symmetric and
it is non-degenerate for semi-simple $\g$. They satisfy the following
compatibility identity
\begin{equation}
	\langle a \otimes [b \otimes c] \rangle
	= \langle [a \otimes b] \otimes c \rangle ,
\end{equation}
where either side defines a trilinear, totally antisymmetric functional.
We extend slightly the notation for $F$-local horizontal forms and write
$\Omega^{h,v}(F,\g)$ for the space of local $\g$-valued variational
$(h,v)$-forms; we also write $\Omega^k(M,\g) \cong \g\otimes
\Omega^k(M)$. The operations $\dv$, $\dh$, ${*}$, $[-]$ and
$\langle-\rangle$ extend to local $\g$-valued variational forms simply
by treating one or the other tensor factor trivially, while $\wedge$ is
extended by acting as $\otimes$ on the Lie algebra factor. Then it is
easy to check that the expressions
\begin{equation}
	\langle \lambda_1 \wedge {*}\lambda_2 \rangle
	\quad\text{and}\quad
	\langle \mu_1 \wedge [\mu_2 \wedge \mu_3] \rangle
\end{equation}
are totally symmetric in their arguments as long as $\mu_i$ are forms of
odd degrees.

\bibliographystyle{utphys}
\bibliography{paper-linstab}

\end{document}